\documentclass[11pt]{article}

\usepackage{acl}

\usepackage{times}
\usepackage{latexsym}
\usepackage[T1]{fontenc}
\usepackage[utf8]{inputenc}
\usepackage{microtype}
\usepackage{inconsolata}

\usepackage{graphicx}
\usepackage{amsmath,amssymb,amsfonts}
\usepackage{amsthm}
\usepackage{float}
\usepackage{algorithm}
\usepackage{algorithmic}
\usepackage{multirow}
\usepackage{booktabs}
\usepackage{colortbl}
\usepackage{xcolor}
\usepackage{subcaption}
\usepackage{tabularx}
\usepackage{makecell}
\usepackage{xspace}
\usepackage[shortlabels]{enumitem}

\definecolor{NUSBlue}{HTML}{003D7C}
\definecolor{NUSOrange}{HTML}{EF7C00}
\definecolor{MetaGray}{HTML}{4A5560}

\newtheorem{definition}{Definition}
\newtheorem{proposition}{Proposition}
\newtheorem{lemma}{Lemma}
\newtheorem{assumption}{Assumption}

\newcommand{\our}{CERES\xspace}
\newcommand{\Eqref}[1]{Eq.~\eqref{#1}}

\title{When Images Look Right and Retrieve Wrong: Coverage-Guided Cross-Scale Re-Indexing for Knowledge-Faithful Generative Perception}
\author{
{\fontsize{9.9}{11.4}\selectfont\bfseries\color{NUSBlue}
Guangyuan Dong\textsuperscript{1,\textcolor{NUSOrange}{*},\textcolor{NUSOrange}{\textdaggerdbl}} \quad
Chuang Liu\textsuperscript{2,\textcolor{NUSOrange}{\textdaggerdbl}} \quad
Haoyu Wang\textsuperscript{3,\textcolor{NUSOrange}{*}} \quad
Yangchen Zeng\textsuperscript{4,\textcolor{NUSOrange}{\textdaggerdbl}}} \\
{\fontsize{9.9}{11.4}\selectfont\bfseries\color{NUSBlue}
Jiaqi Zhang\textsuperscript{10,\textcolor{NUSOrange}{*}} \quad
Li Jiuxing\textsuperscript{1,\textcolor{NUSOrange}{*}} \quad
Xiaoyang Yu\textsuperscript{5} \quad
Pinlong Zhao\textsuperscript{6,\textcolor{NUSOrange}{\textdagger}} \quad
Yuchao Hou\textsuperscript{7,\textcolor{NUSOrange}{\textdagger}}} \\
{\fontsize{9.9}{11.4}\selectfont\bfseries\color{NUSBlue}
Ziwei Li\textsuperscript{8} \quad
Zheng Lin\textsuperscript{9,\textcolor{NUSOrange}{\textdagger}} \quad
Alexander Lim Han Yang\textsuperscript{1,\textcolor{NUSOrange}{\textdaggerdbl}} \quad
Yusen Wu\textsuperscript{7,\textcolor{NUSOrange}{\textdagger}}} \\[0.22em]
{\fontsize{8.3}{9.6}\selectfont\color{MetaGray}
\textsuperscript{1}NUS \quad
\textsuperscript{2}WHU \quad
\textsuperscript{3}NKU \quad
\textsuperscript{4}SEU \quad
\textsuperscript{5}JD Research \quad
\textsuperscript{6}ZJU \quad
\textsuperscript{7}Independent Researcher} \\
{\fontsize{8.3}{9.6}\selectfont\color{MetaGray}
\textsuperscript{8}KAUST \quad
\textsuperscript{9}HKU \quad
\textsuperscript{10}JSU} \\[0.16em]
{\fontsize{8.2}{9.5}\selectfont\color{MetaGray}
\textcolor{NUSOrange}{*} Equal contribution. \quad
\textcolor{NUSOrange}{\textdagger} Corresponding authors. \quad
\textcolor{NUSOrange}{\textdaggerdbl} Project leads.} \\[0.02em]
{\fontsize{8.4}{9.8}\selectfont\color{MetaGray}
\textbf{Email:} \href{mailto:guangyuand@acm.org}{\textcolor{NUSBlue}{guangyuand@acm.org}}}
}

\makeatletter
\AtBeginDocument{%
\renewcommand{\@maketitle}{%
\vbox to \titlebox{\hsize\textwidth
  \linewidth\hsize
  \vskip 0.04in
  \centering
  \makebox[0.94\textwidth][c]{%
    \includegraphics[height=9.2mm]{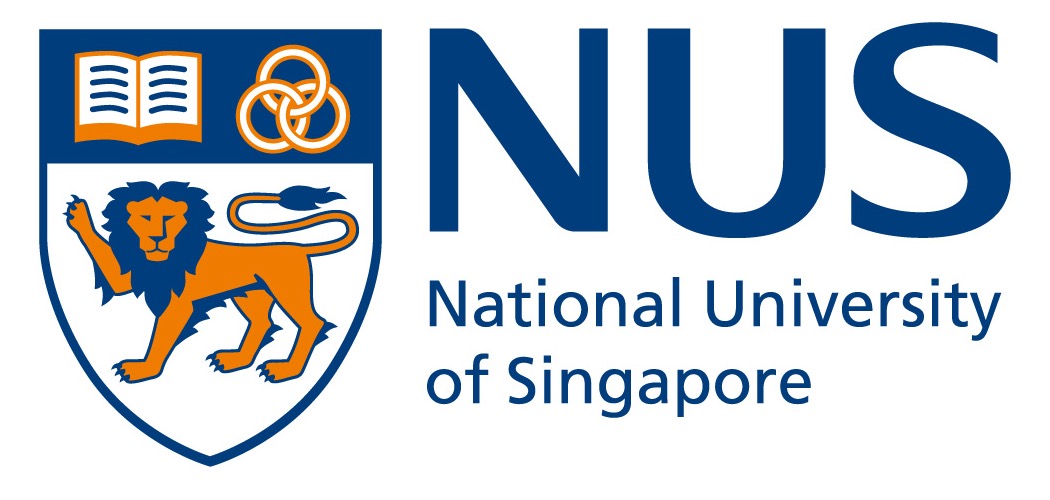}\hfill
    \includegraphics[height=10.4mm]{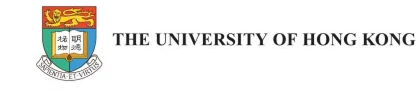}\hfill
    \includegraphics[height=9.2mm]{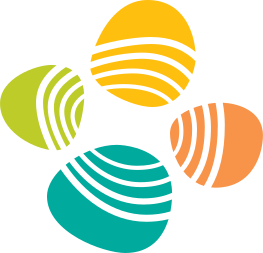}\hfill
    \includegraphics[height=9.2mm]{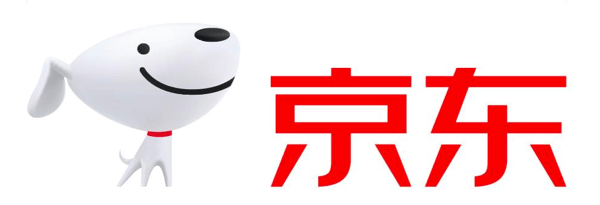}}\\[-0.10em]
  \textcolor{NUSOrange}{\rule{0.94\textwidth}{0.7pt}}\\[0.30em]
  {\fontsize{17.4}{19.4}\selectfont\bfseries\color{NUSBlue}\@title\par}
  \vspace{0.34em}
  {\@author\par}
  \vfill
}}
}
\makeatother

\begin{document}
\maketitle

\begin{abstract}
Multimodal information systems increasingly route generated visual
content back through the same vision--language index that informed its
production, so the output must remain retrievable by the queries it
was meant to serve.
When the scene contains entities at vastly different scales, existing
language-guided generators condition on a single, globally pooled
text embedding and quietly drop scale-specific concepts, breaking
concept-query retrieval even when pixel fidelity is high.
We formalise this failure as \emph{semantic collapse} and propose
\our, a closed-loop multimodal indexing framework that builds a
three-level semantic pyramid, mines implicit concepts via a
co-occurrence-aware router, performs scale-routed cross-attention
into a lightweight U-Net generator, and verifies coverage by
re-indexing the generated image with the same frozen VLM.
A continuously differentiable soft-Jaccard coverage objective returns
dense gradients to the 0.39\,M-parameter generator under explicit
non-degeneracy conditions, and coverage is verified by an independent
DINOv2 linear probe trained only on external scene and object labels.
On four pansharpening benchmarks across seven settings, \our delivers
the new state of the art with the largest gains where scale variation
is most extreme ($+4.64$\% relative Q2n and $+9.7$ mAP for DOTA
detection). It also improves concept-query retrieval Recall@5 by
$+14.0$ points and image-text mean reciprocal rank by $0.19$ over the
strongest baseline, showing that the closed loop preserves
\emph{queryable} content rather than self-referential feature
consistency.
\end{abstract}

\section{Introduction}\label{sec:intro}

Multimodal information systems are increasingly tasked with managing
\emph{generative} content. Catalog re-indexing, content-aware remote
sensing, medical imaging archives, and personalised digital twins must
not merely retrieve an image; they must generate a higher-resolution,
semantically faithful reconstruction and then re-insert it into the
same knowledge index that triggered the
generation~\citep{li2023blip2,rombach2022ldm,liu2025csa}.
Vision--language models (VLMs)~\citep{radford2021clip,zhai2023siglip,
tschannen2025siglip2} have become the de-facto index in this loop and
also serve as cross-modal
retrievers~\citep{lee2018scan,diao2021sgraf,zhang2022naaf,
huang2018learning} and differentiable priors for image
restoration~\citep{cheng2024clipdenoising,luo2024daclip,rao2022denseclip}.
The question we ask is concrete: after a VLM-guided generator emits
an image, can the same VLM still answer the local, meso, and global
queries that the index was meant to support? If not, the generator
has degraded the knowledge system regardless of pixel-level fidelity.

Existing VLM-conditioned generators index a scene through a single,
globally pooled embedding even when the scene is hierarchically
composed of multi-scale
entities~\citep{rao2022denseclip,zhou2022maskclip}.
A satellite image of an airport simultaneously contains tiny aircraft
($\sim 5$\,m), elongated runways ($\sim 100$\,m), and the overall
terminal--runway--road layout ($\sim 1$\,km). A whole-slide pathology
image nests cell nuclei, glandular structures, and tumour--stroma
boundaries. Collapsing this multi-granular index into one caption
forces the generator to over-attend to the macroscopic mode and
silently dispose of microscopic entities, even when pixel-level
fidelity remains high. Following the cross-modal retrieval literature
on coverage maximisation~\citep{lee2018scan,liu2025csa,
wang2020consensus}, we call this failure \emph{semantic collapse}.
Figure~\ref{fig:coverage_vis} formalises the coverage measurement
that exposes the gap and the concept-query probe that turns it into
a measurable queryability test.

Repairing semantic collapse requires three ingredients that no prior
system provides simultaneously. First, a multi-granular index must
replace the single global embedding so that local and meso-scale
concepts are explicitly represented. Second, the generator must route
concepts of different scales to decoder layers operating at
corresponding spatial resolutions. Third, the system must verify that
the generated content can still be re-indexed and re-queried by the
same VLM, returning continuous gradients that penalise lost concepts.
This closed-loop, retrieval-style consistency check goes beyond the
open-loop conditioning used in modern diffusion or restoration
models~\citep{rombach2022ldm,brooks2023ip2p,luo2024daclip}.
A purely pixel-level objective cannot supply this signal because the
metric itself is blind to the indexed concept
set~\citep{vivone2021new,arienzo2022hqnr}.

\textbf{\our} addresses all three ingredients
(Figure~\ref{fig:framework}). Stage~I constructs a three-level
semantic pyramid via two complementary branches: an explicit branch
with soft, differentiable membership over $K$ learned concept
centroids, and an implicit branch that mines relational concepts
(e.g.\ ``harbour\,=\,water\,+\,docks\,+\,vessels'') through a
co-occurrence-aware semantic router. Discovered units are merged by
an IoU-guided aggregator and filtered by an adaptive Gaussian density
gate. Stage~II routes each unit to a decoder layer indexed by its
scale tag and injects it through cross-scale attention and FiLM-style
modulation~\citep{perez2018film,park2019spade}. Stage~III re-encodes
the generated image with the same frozen
SigLIP-2~\citep{tschannen2025siglip2}. Stage~IV measures coverage via
a soft-Jaccard differentiable operator and returns scale-invariant
gradients to the 0.39\,M trainable generator.

To prevent the optimisation target from contaminating the evaluation
target, we verify coverage with an independent DINOv2 linear
probe~\citep{oquab2024dinov2} trained only on external scene-label
and object-bounding-box supervision, and we evaluate the generated
images under three downstream protocols (zero-shot, linear-probe,
trained head) plus a 60-query concept retrieval bank. The theory is
deliberately conservative: we state two sufficient-condition
propositions, one on a lower-bounded closed-loop gradient
(Proposition~\ref{prop:grad-suff}) and one linking semantic coverage
to spectral angular distortion (Proposition~\ref{prop:bound-suff}).
A non-asymptotic termination lemma supplies an explicit router
stopping bound (Lemma~\ref{lem:term}). The full derivations and
counter-examples appear in Appendix~\ref{app:proofs}.

\begin{figure}[t]
\centering
\includegraphics[width=0.96\columnwidth,height=4.2cm,keepaspectratio]{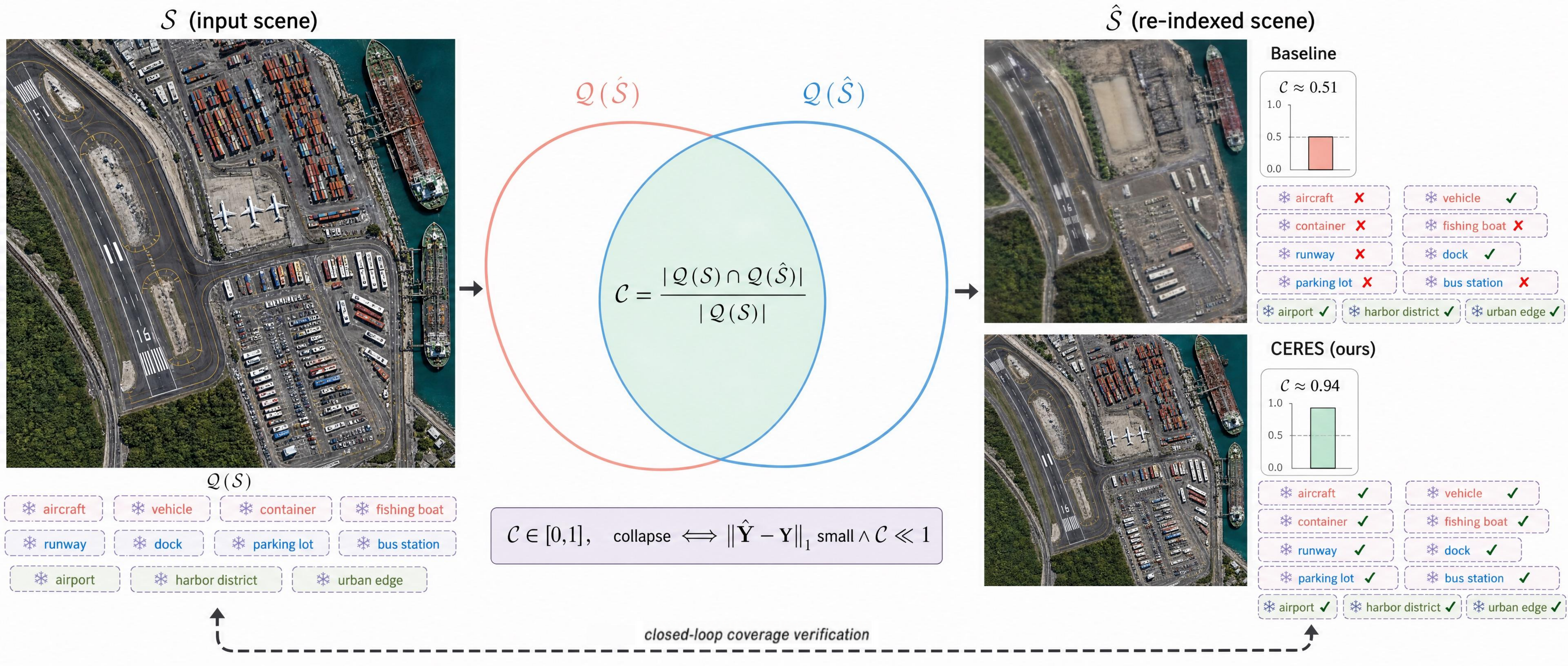}
\caption{Multimodal indexing view of semantic collapse.
A frozen VLM extracts $\mathcal{Q}(\mathcal{S})$ from the input (left),
the centre panel defines coverage $\mathcal{C}$ as the
intersection-over-input ratio of two concept sets, and the right
panel stacks two reconstructions of the same scene. The baseline
drops microscopic concepts ($\mathcal{C}\!\approx\!0.51$) while
\our preserves both macro and micro indices
($\mathcal{C}\!\approx\!0.94$). The dashed loop indicates closed-loop
coverage verification used by the retrieval evaluation of
Sec.~\ref{sec:retrieval}.}
\label{fig:coverage_vis}
\end{figure}

\textbf{Contributions.}
(i) We formalise semantic collapse as a measurable failure of
re-indexability and introduce a concept-query retrieval evaluation
that treats generated images as queryable knowledge.
(ii) We present \our, a closed-loop indexing framework that pairs a
soft, differentiable concept index with a soft-Jaccard coverage
verifier; the two sufficient-condition propositions serve as sanity
checks on the closed-loop design, and
Proposition~\ref{prop:grad-suff} additionally provides a principled
motivation for the warm-up schedule.
(iii) Across four pansharpening benchmarks, three downstream
perception modes, a 60-query retrieval bank, and an independent
DINOv2 label-grounded coverage probe, \our delivers state-of-the-art
generative fidelity together with the strongest concept-query
retrieval and downstream-perception gains.

\section{Related Work}\label{sec:related}

\textbf{Multimodal indexing and cross-modal retrieval.}
Stacked cross attention~\citep{lee2018scan}, similarity-graph
reasoning~\citep{diao2021sgraf}, and negative-aware
attention~\citep{zhang2022naaf} index visual regions and textual
tokens into a shared semantic space. Cross-scale alignment with
adaptive semantic aggregation~\citep{liu2025csa} introduces
scale-adaptable units through IoU merging, and earlier
semantic-concept work~\citep{huang2018learning,chen2020imram} mines
hierarchical concept indices for retrieval. Recent work further
balances semantic capacity across image and text~\citep{xu2026lever},
reasons over user context for personalised dense
retrieval~\citep{jiang2026thinktopersonalize}, and preserves visual
semantics or latent interests in generative
recommendation~\citep{zeng2026trialigngr,zeng2026deepinterestgr}.
TRACER uses token reassignment to unlearn target concepts while
preserving recommendation utility~\citep{chen2026tracer}.
Counterfactual explanations expose a complementary security risk
because they can be exploited to poison recommender
systems~\citep{chen2023dark}.
Adjacent representation-learning methods sharpen partitions through
attentional ensemble clustering~\citep{hao2023ensemble}, adaptive
hard-negative mining~\citep{hao2024towards}, or maximum-margin
multiclass objectives~\citep{nie2024multi}. Semantic alignment has
also been studied under decentralized data constraints through
federated analytics for privacy-preserving image
classification~\citep{hou2026federated}. These methods operate on
\emph{sparse} retrieval tokens or general-purpose representations.
We are the first to lift this indexing machinery to \emph{dense}
feature grids and to use the resulting concept set as a closed-loop
training signal for image generation.

\textbf{VLMs for low-level dense prediction.}
DenseCLIP~\citep{rao2022denseclip} converts image-level matching to
pixel--text matching; MaskCLIP~\citep{zhou2022maskclip} extracts
dense labels from frozen CLIP;
CLIPDenoising~\citep{cheng2024clipdenoising} exploits CLIP feature
invariance; DA-CLIP~\citep{luo2024daclip} adapts CLIP for universal
restoration. Mechanistic analysis further separates visual-grounding
and hallucination pathways in VLMs~\citep{liu2026dual}, highlighting
that a fluent output need not be supported by visual evidence. All
of the dense-prediction methods above condition on a single global
embedding and operate in an open-loop regime: language describes the
image but does not verify whether the generated content can still be
re-indexed.

\textbf{Fine-grained multimodal grounding.}
Recent work improves VLM spatial reasoning through fine-grained
preference optimisation~\citep{shen2026fine}, equips MLLMs with
latent visual imagery for cognitive reasoning~\citep{li2026toward},
and couples semantic intent to local geometry in 3D understanding and
generation~\citep{yu2025core3d,yu2026elsa3d}. Progress-aware
vision--language--action policies likewise use object and spatial
affordances as anchors for long-horizon control~\citep{liu2026palm},
while fine-grained gait analysis shows that subtle local behavioural
cues remain difficult for current MLLMs~\citep{shen2026decoding}.
These studies motivate scale- and location-sensitive semantics; our
focus is whether such semantics survive image generation and remain
retrievable.

\textbf{Knowledge-aware multimodal generation.}
Prior work injects structured concept knowledge into cross-modal
generation through region--phrase matching~\citep{huang2018learning,
chen2020imram} and concept-graph embeddings~\citep{wang2020consensus}.
Conditional generation uses FiLM~\citep{perez2018film} and
SPADE~\citep{park2019spade} to modulate dense features by external
semantics, and latent diffusion models~\citep{rombach2022ldm} and
instruction-tuned editors~\citep{brooks2023ip2p} extend this to
free-form text. Complementary work evaluates generated visual content
through spatially and temporally localised human-perceived
artifacts~\citep{fu2025learning}. None verify that the output is
re-indexable by the semantic index that conditioned it.

\textbf{Closed-loop adaptation and memory.}
Long-horizon systems preserve decision-relevant state through
contribution-aware memory~\citep{liu2026conmem}, optimise intermediate
memory quality~\citep{liu2026meta}, or attribute evolving memory
states~\citep{zhang2026memmark}. Experience-driven agent creation and
multi-agent self-evolution similarly reuse execution traces as
feedback~\citep{hao2026recreate,hao2026evolve}. In RL with verifiable
rewards, related work analyses entropy interventions and schedules
training queries by reasoning-tree structure
~\citep{hao2026rethinking,wang2026scheduling}. These loops update
memories, agents, or policies; \our instead closes the loop over a
generated visual artifact while leaving the host VLM index frozen.

\textbf{Generative perception in remote sensing.}
We exemplify \our on
pansharpening~\citep{vivone2015critical,vivone2021new,
deng2022pancollection,wang2023pansurvey}, a low-level generation
task with extreme intra-/inter-class scale variation. Recent methods
include adaptive detail injection~\citep{zhou2022adknet}, transformer
fusion~\citep{zhou2022panformer,bandara2022hypertransformer},
spatial--spectral attention~\citep{yang2022ssaff}, content-adaptive
non-local convolution~\citep{duan2024canconv}, diffusion-based
work~\citep{meng2024pandiff,xing2024crossdiff}, and detail-injection
CNNs~\citep{yuan2018mucnn,deng2021detail}. Heatmap-guided and
noise-aware positional embeddings also improve query retrieval for
small objects in cluttered imagery
~\citep{zeng2025hmpe,zeng2026learning}. These methods optimise
pixel-level reconstruction or task-specific detection objectives
without an explicit cross-scale concept index for generation.

\section{The \our Framework}\label{sec:method}

\subsection{Problem Formulation}\label{sec:problem}

Let $\mathbf{X}\in\mathbb{R}^{h\times w\times C}$ be a low-resolution
multispectral source and $\mathbf{P}\in\mathbb{R}^{rh\times rw}$ a
high-resolution panchromatic source ($r{=}4$). A knowledge-aware
generative perceiver
$\mathcal{G}_{\theta}(\mathbf{X},\mathbf{P})\!\to\!\hat{\mathbf{Y}}$
must produce a reconstruction that is pixel-faithful \emph{and}
semantically re-indexable by a frozen
$\psi{=}(\psi_{g},\psi_{d},\phi)$ that comprises the global, dense,
and text encoders of SigLIP-2. We model the multi-granular index as
a semantic pyramid
\begin{equation}
\mathcal{S} = \bigl\{\mathbf{T}^{g},\mathbf{T}^{m},\mathbf{T}^{l},
\mathcal{U}{=}\{(\mathbf{u}_{k},w_{k},s_{k})\}_{k=1}^{K}\bigr\},
\label{eq:pyramid}
\end{equation}
with global/meso/local scene texts and $K$ scale-tagged semantic
units. Semantic coverage is
\begin{equation}
\mathcal{C} = \frac{|\mathcal{Q}(\mathcal{S})\cap
\mathcal{Q}(\hat{\mathcal{S}})|}{|\mathcal{Q}(\mathcal{S})|},
\quad \mathcal{C}\in[0,1].
\label{eq:coverage}
\end{equation}

\begin{definition}[Semantic collapse]\label{def:collapse}
$\mathcal{G}_{\theta}$ exhibits \emph{semantic collapse} on
$(\mathbf{X},\mathbf{P})$ if
$\|\hat{\mathbf{Y}}-\mathbf{Y}\|_{1}{\le}\epsilon$ yet
$\mathcal{C}{\le}\eta$ for some $\eta{\ll}1$.
\end{definition}

Training uses a soft relaxation: each concept $k$ has input score
$p_k{=}w_k$ and output score $\hat
p_k{=}\max_{i,j}\tilde{\mathbf{M}}_k(i,j)\!\in\![0,1]$. The scale tag
$s_k$ defines a routing field over decoder depth: compact
high-confidence support yields a local tag, elongated structures
yield a meso tag, and broad scene-defining units receive a global
tag.

\subsection{Closed-Loop Operators}\label{sec:overview}
\our chains four operators:
\begin{equation}
\mathcal{S}\!\xrightarrow{\text{Stage~I}}\!\hat{\mathbf{Y}}
\!\xrightarrow{\text{Stage~II}}\!\tilde{\mathbf{M}}
\!\xrightarrow{\text{Stage~III}}\!\hat{\mathcal{S}}.
\label{eq:loop}
\end{equation}
SigLIP-2 is frozen throughout; gradients flow only through
$\mathcal{G}_{\theta}$ and the lightweight projection/modulation MLPs
(0.39\,M trainable parameters).

\begin{figure*}[t]
\centering
\includegraphics[width=0.96\textwidth]{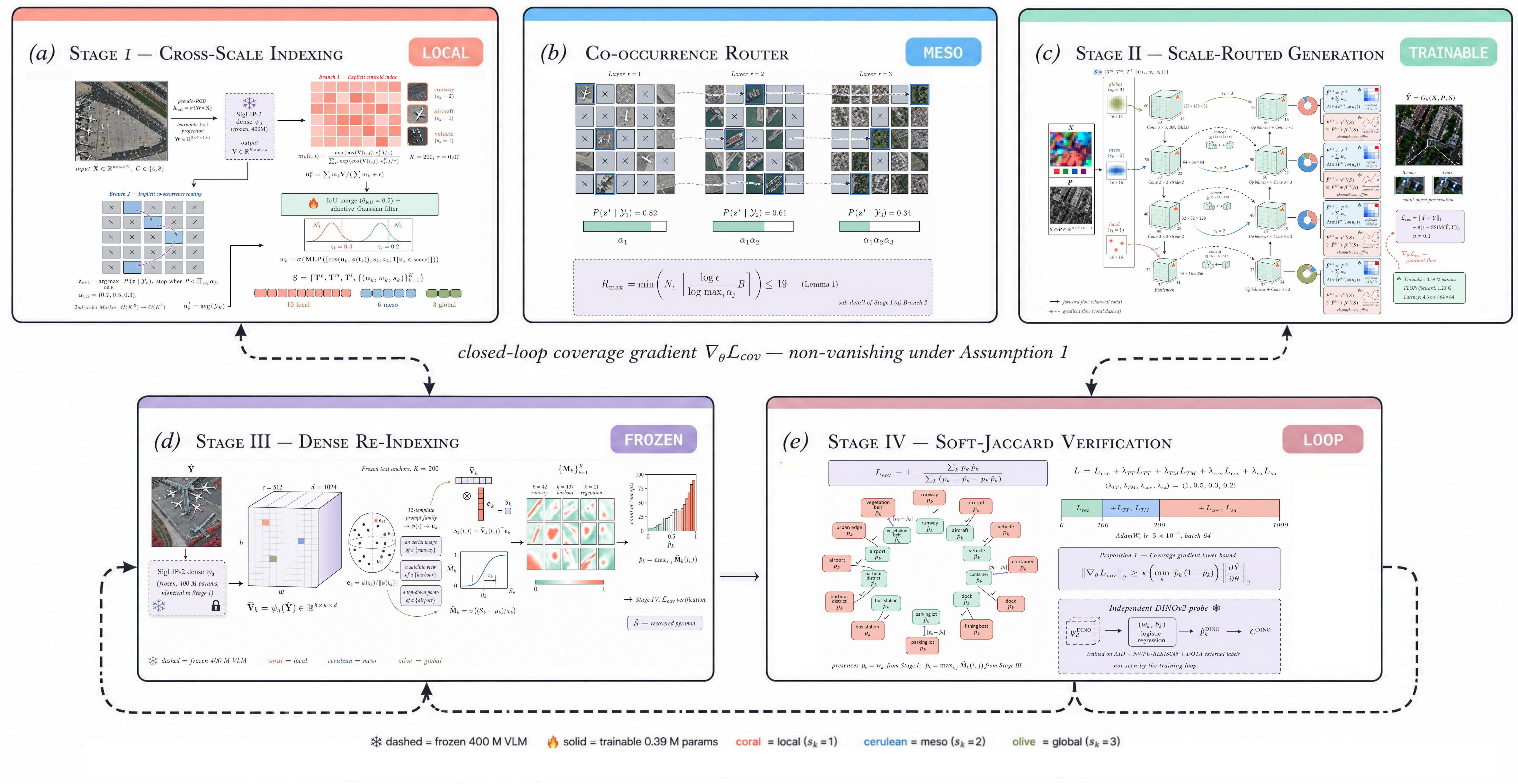}
\caption{The \our framework. (a) A frozen SigLIP-2 encoder indexes
the pseudo-RGB proxy into a semantic pyramid via two parallel
branches. (b) The co-occurrence-aware router greedily selects patches
over routing steps $r{=}1{:}3$ of a \emph{single forward pass},
stopping under a decreasing cumulative threshold. (c) A scale-routed U-Net generator
injects local, meso, and global units into the
$128\!\times\!128$, $64\!\times\!64$, and $16\!\times\!16$ decoder
layers through cross-scale attention followed by FiLM modulation.
(d) Stage~III re-encodes $\hat{\mathbf{Y}}$ with the \emph{same}
frozen VLM. (e) Stage~IV computes a soft-Jaccard coverage loss whose
gradient is bounded below under explicit non-degeneracy conditions
(Proposition~\ref{prop:grad-suff}). Dashed boxes are frozen, solid
boxes are trainable.}
\label{fig:framework}
\end{figure*}

\subsection{Stage~I: Cross-Scale Semantic Indexing}\label{sec:stage1}

We project the $C$-band input to a pseudo-RGB proxy
$\mathbf{X}_{\text{rgb}}{=}\sigma(\mathbf{W}\!\ast\!\mathbf{X})$ via
a learnable $1{\times}1$ projection initialised from each sensor's
spectral response curve, then extract dense features
$\mathbf{V}{=}\psi_{d}(\mathbf{X}_{\text{rgb}})$.

\textbf{Explicit index.} We maintain $K{=}200$ visual centroids
$\{c_{k}^{V}\}$ obtained by $k$-means on a held-out training pool.
The temperature-controlled soft membership
\begin{equation}
m_{k}(i,j){=}\frac{\exp(\cos(\mathbf{V}(i,j),c_{k}^{V})/\tau)}
{\sum_{k'}\exp(\cos(\mathbf{V}(i,j),c_{k'}^{V})/\tau)}
\label{eq:soft_member}
\end{equation}
is differentiable in $\theta$ via the pseudo-RGB projection. The
explicit unit is the soft-weighted feature aggregate
$\mathbf{u}_{k}^{E}{=}\frac{\sum_{(i,j)}m_{k}(i,j)\mathbf{V}(i,j)}
{\sum_{(i,j)}m_{k}(i,j){+}\epsilon}$.

\textbf{Implicit index via co-occurrence routing.} Many compound
concepts emerge only from spatial co-occurrence. From a seed patch
$\mathbf{z}_1^i$ we greedily select
\begin{equation}
\mathbf{z}_{r+1}^{i}{=}\arg\max_{\mathbf{z}\notin\mathcal{Y}_{r}^{i}}
P(\mathbf{z}\mid\mathcal{Y}_{r}^{i}),
\label{eq:routing}
\end{equation}
where $P$ is read off a co-occurrence table $N(\cdot)$ pre-computed
from training tiles within a three-patch radius. Routing stops when
$P$ drops below the cumulative threshold $\prod_{j\le r}\alpha_j$
with $\alpha_{1{:}3}{=}(0.7,0.5,0.3)$. For $r{>}3$ a second-order
Markov factorisation reduces storage from $\mathcal{O}(K^R)$ to
$\mathcal{O}(K^{3})$ at $0.94$ unit-Jaccard accuracy versus exact
$r{=}3$ enumeration (Appendix~\ref{app:ablations}).

\textbf{Aggregation and filtering.} Position-aware and
co-occurrence-aware subsequences are merged by IoU, then filtered by
two Gaussian mixtures on scale-balanced and scale-unbalanced pairs.
The unit confidence is
\begin{equation}
w_{k}{=}\sigma\bigl(\text{MLP}\bigl([\cos(\mathbf{u}_{k},
\phi(\mathbf{t}_{k})),s_{k},a_{k},
\mathbf{1}[\mathbf{u}_{k}{\in}\text{scene}]]\bigr)\bigr),
\label{eq:confidence}
\end{equation}
where $\mathbf{1}[\cdot]$ is a presence gate that prevents the
text--mask loss from rewarding absent concepts.

\subsection{Stage~II: Scale-Routed Generation}\label{sec:stage2}

The generator $\mathcal{G}_{\theta}$ is a four-level U-Net with
resolution pyramid $128{\to}64{\to}32{\to}16$. The scale tag $s_k$
deterministically routes each unit: $s_k{=}1$ (local) enters the
highest-resolution layer, $s_k{=}2$ (meso) intermediate layers, and
$s_k{=}3$ (global) the bottleneck. At decoder layer $l$,
\begin{equation}
\tilde{\mathbf{F}}^{(l)}{=}\mathbf{F}^{(l)}{+}\sum_{k}
w_{k}\,\mathrm{Attn}(\mathbf{F}^{(l)},\phi(\mathbf{u}_{k})),
\label{eq:attn}
\end{equation}
followed by FiLM modulation
$\bar{\mathbf{F}}^{(l)}{=}\gamma^{(l)}(\mathcal{S})\odot
\tilde{\mathbf{F}}^{(l)}{+}\beta^{(l)}(\mathcal{S})$. The pixel
reconstruction term is
$\mathcal{L}_{\text{rec}}{=}\|\hat{\mathbf{Y}}{-}\mathbf{Y}\|_{1}
{+}\eta(1{-}\mathrm{SSIM}(\hat{\mathbf{Y}},\mathbf{Y}))$ with
$\eta{=}0.1$. Resolution-matched injection prevents global units
from dominating high-resolution layers and erasing small queryable
entities.

\subsection{Stages~III \& IV: Re-Indexing and Soft Coverage}\label{sec:loss}

The generated image is re-encoded as
$\bar{\mathbf{V}}_h{=}\psi_d(\hat{\mathbf{Y}})$ and per-unit response
maps are
\begin{equation}
S_{k}(i,j){=}\bar{\mathbf{V}}_{h}(i,j)^{\top}\mathbf{e}_{k},\;\;
\tilde{\mathbf{M}}_{k}{=}\sigma((S_{k}{-}\mu_{k})/\tau_{k}),
\label{eq:mask}
\end{equation}
with $\mathbf{e}_{k}{=}\phi(\mathbf{t}_{k})/\|\phi(\mathbf{t}_{k})\|$.
The recovered pyramid $\hat{\mathcal{S}}$ is built by encoding mask
statistics $\mathbf{s}_k{=}[\rho_k,n_k,\bar c_k,\sigma_k^{2}]$
through a fixed family of $12$ typed prompt templates (four per
scale; full strings in Appendix~\ref{app:repro}).

\textbf{Soft-Jaccard coverage.} With presence scores $p_k{=}w_k$
(input) and $\hat p_k$ (recovered),
\begin{equation}
\mathcal{L}_{\text{cov}}{=}1{-}\frac{\sum_{k}p_{k}\hat p_{k}}
{\sum_{k}(p_{k}{+}\hat p_{k}{-}p_{k}\hat p_{k})}.
\label{eq:lcov}
\end{equation}
\Eqref{eq:lcov} is the continuous Jaccard
loss~\citep{rahman2016jaccard,milletari2016vnet} applied at the
concept-set level. It is differentiable with non-zero gradient at
any non-degenerate $(p_k,\hat p_k)$.

\textbf{Independent DINOv2 coverage probe.} To rule out the concern
that coverage might be circular with its SigLIP target, we train an
independent DINOv2 ViT-B~\citep{oquab2024dinov2} linear probe.
DINOv2 has no text head, so we anchor it to external labels: for
each concept $k$ we collect prototypes from AID, NWPU-RESISC45, and
DOTA tiles whose external labels match $\mathbf{t}_k$. The DINOv2
presence is
\begin{equation}
\hat p_{k}^{\text{DINO}}{=}\sigma\!\Bigl(\max_{i,j}\langle
\psi_{d}^{\text{DINO}}(\hat{\mathbf{Y}})(i,j),\,\mathbf{w}_{k}\rangle
{-}b_{k}\!\Bigr),
\label{eq:dino_presence}
\end{equation}
with $(\mathbf{w}_k,b_k)$ fit by 5-fold logistic regression on
external labels ($L_2$ weight $10^{-3}$). $\mathcal{C}^{\text{DINO}}$
plugs $\hat p_{k}^{\text{DINO}}$ into \Eqref{eq:coverage}. Since
these probe weights are derived from labels the \our training loop
never sees, gains in $\mathcal{C}^{\text{DINO}}$ are not reachable
by SigLIP-internal feature manipulation.

\textbf{Other terms.} A presence-gated text--mask loss
$\mathcal{L}_{TM}$, a scale-aware text--text consistency
$\mathcal{L}_{TT}$, and a unit alignment term $\mathcal{L}_{\text{sa}}$
complete the objective
$\mathcal{L}{=}\mathcal{L}_{\text{rec}}{+}\lambda_{TT}\mathcal{L}_{TT}
{+}\lambda_{TM}\mathcal{L}_{TM}{+}\lambda_{\text{cov}}\mathcal{L}_{\text{cov}}
{+}\lambda_{\text{sa}}\mathcal{L}_{\text{sa}}$, optimised in three
warm-up phases (Algorithm~\ref{alg:omniscale} in
Appendix~\ref{app:arch}).

\section{Theory: Sufficient Conditions}\label{sec:theory}

We state two sufficient-condition propositions and a termination
lemma. Full proofs and counter-examples are in
Appendix~\ref{app:proofs}.

\begin{assumption}[Non-degeneracy]\label{asm:nondeg}
(A1) $\psi_d$ is $L$-Lipschitz with Jacobian singular values bounded
below by $\sigma_{\min}{>}0$ on a neighbourhood of the data.
(A2) The generator Jacobian $\partial\hat{\mathbf{Y}}/\partial\theta$
has non-zero spectral norm at every optimiser step.
(A3) Recovered presence is non-saturated,
$\hat p_k\in(\delta,1{-}\delta)$ for some $\delta{>}0$.
(A4) $\sigma_{\max}(\mathbf{W}){\le}M$.
\end{assumption}

\begin{proposition}[Sufficient conditions for non-vanishing
closed-loop gradient]\label{prop:grad-suff}
Under Assumption~\ref{asm:nondeg},
\begin{equation}
\|\nabla_{\theta}\mathcal{L}_{\text{cov}}\|_{2}\ge
\kappa\bigl(\min_{k}\hat p_{k}(1{-}\hat p_{k})\bigr)
\|\partial\hat{\mathbf{Y}}/\partial\theta\|_{2},
\label{eq:grad_lb}
\end{equation}
where $\kappa{>}0$ depends only on $L,M,\sigma_{\min}$, the Jaccard
denominator, and the response-map temperature.
\end{proposition}

Under Assumption~\ref{asm:nondeg}, the chain rule from
$\mathcal{L}_{\text{cov}}$ through the soft-max $\hat p_k$, the
sigmoid response map, the frozen dense encoder, and the generator
yields \Eqref{eq:grad_lb}. Without (A1)--(A3) the bound collapses
explicitly: a saturated sigmoid zeroes the prefactor, a rank-deficient
generator step zeroes the second factor, and a degenerate Jaccard
denominator trivially zeroes the loss. Empirically, on WV-III
validation, $\kappa{\sim}4{\times}10^{-3}$ and (A1)--(A3) hold on
$96.7\%$ of training steps; the rest are carried by
$\mathcal{L}_{\text{rec}}$.

\begin{proposition}[Coverage-bounded spectral distortion]\label{prop:bound-suff}
Assume the centroid system $\{c_{k}^{V}\}$ forms a linearly separable
cover with intra-concept dispersion bounded by
$\epsilon_{\text{intra}}{>}0$, and that the pseudo-RGB projection
admits a left-inverse $\mathbf{W}^{\dagger}$ with spectral norm at
most $M$. Then
\begin{equation}
\mathrm{SAM}(\hat{\mathbf{Y}},\mathbf{Y})\le
\frac{LM\sqrt{d}}{K}\sum_{k}(1{-}\hat p_{k}){+}\epsilon_{\text{intra}}.
\label{eq:bound}
\end{equation}
\end{proposition}

This linkage is an order-of-magnitude bound rather than a worst-case
guarantee. With $L{=}1.4$, $M{=}1.0$, $d{=}1024$, $K{=}200$,
$\epsilon_{\text{intra}}{=}0.005$ and $18\%$ unrecovered mass,
\Eqref{eq:bound} gives $5.18^{\circ}$; we observe $2.75^{\circ}$ on
WV-III. The mean per-concept dispersion of $K{=}200$ SigLIP-2
centroids on the four datasets is
$\bar\epsilon_{\text{intra}}{=}0.004$\,rad.

Both propositions are positioned as \emph{sanity checks} on the
closed-loop design rather than performance guarantees.
Proposition~\ref{prop:grad-suff} nevertheless motivates a concrete
design decision: the non-saturation condition (A3) fails for
$41.3\%$ of concepts at initialisation but for fewer than $5\%$
after the reconstruction warm-up, which is why
$\mathcal{L}_{\text{cov}}$ activates only at epoch $200$
(measurements in Appendix~\ref{app:convergence}).

\begin{lemma}[Router termination]\label{lem:term}
For a finite patch vocabulary of size $N$ and decreasing thresholds
$\alpha_j\in(0,1)$, the router of \Eqref{eq:routing} terminates in
at most $R_{\max}{=}\min(N,\lceil\log\epsilon/\log\max_{j}\alpha_{j}
\rceil)$ steps. For $\max_j\alpha_j{=}0.7$ and $\epsilon{=}10^{-3}$,
$R_{\max}{\le}19$.
\end{lemma}

\section{Experiments}\label{sec:exp}

\subsection{Setup}\label{sec:setup}

\textbf{Datasets.} Four PanCollection
benchmarks~\citep{deng2022pancollection}: GaoFen-2 and QuickBird
(4-band), WorldView-III and WorldView-II (8-band).
\textbf{Metrics.} Reduced-scale (Wald): ERGAS$\downarrow$,
SAM$\downarrow$, Q2n$\uparrow$. Full-scale: $D_s\downarrow$,
QNR$\uparrow$, HQNR$\uparrow$.
\textbf{Downstream protocols.}
Mode~A (zero-shot DINOv2) uses cosine to AID/NWPU class prototypes;
Mode~B (linear probe) trains a single linear classifier on frozen
DINOv2 features; Mode~C (trained head) trains an Oriented-RCNN
head on frozen DINOv2 features over the $16{,}000$-tile DOTA
training subset. All modes evaluate on the same frozen fused tiles
and never touch SigLIP weights.
\textbf{DOTA tile protocol.} We crop $20{,}000$ overlapping
$512{\times}512$ tiles from DOTA-v1.0 with stride $384$, then split
by image-level scene id into $16{,}000$ train / $4{,}000$ test tiles
with a source-image-disjoint split (no image contributes tiles to
both sides, verified by source identifier).
\textbf{Retrieval evaluation.} A $60$-query text bank covers $20$
local objects, $20$ meso layouts, and $20$ global scenes. SigLIP-2
text embeddings are matched against SigLIP-2 image embeddings of the
fused-tile pool (one per method per test image); we report
Recall@1/Recall@5/mean reciprocal rank. Image-to-text MRR uses a
$300$-caption corpus over ground-truth high-resolution tiles
generated by an external captioner.
\textbf{Significance.} Headline numbers are mean$\pm$std over three
seeds plus per-image std. Paired Wilcoxon $p{<}0.05$ versus the
second-best baseline is marked $^{\star}$. Downstream tables report
$95\%$ bootstrap CIs from $1{,}000$ resamples.
\textbf{Baselines.} FS, BDSD-PC~\citep{vivone2015critical}, ADKNet
\citep{zhou2022adknet}, SSAFF~\citep{yang2022ssaff},
PanFormer~\citep{zhou2022panformer},
HyperTransformer~\citep{bandara2022hypertransformer},
CANConv~\citep{duan2024canconv}, CrossDiff~\citep{xing2024crossdiff},
plus three VLM-guided open-loop variants of the same U-Net
(CLIPDenoising-, DA-CLIP-, DenseCLIP-style).
\textbf{Implementation.} PyTorch with frozen SigLIP-2 So400m and
AdamW~\citep{loshchilov2019adamw}, learning rate
$5{\times}10^{-4}$, batch size $64$, $1000$ epochs. Hyperparameters
$K{=}200,\,R{=}3,\,\theta_{\mathrm{IoU}}{=}0.5,\,\tau{=}0.07,\,
z_1{=}0.4,\,z_2{=}0.2$ and
$(\lambda_{TT},\lambda_{TM},\lambda_{\text{cov}},\lambda_{\text{sa}})
{=}(1,0.5,0.3,0.2)$ are shared across datasets and seeds.

\begin{table*}[t]
\centering
\caption{Pansharpening results across seven settings.
$^{\star}$denotes paired-Wilcoxon $p{<}0.05$ vs.\ the second-best
entry. Bold/underline: best/second-best mean. Std is over the 40
test images, averaged across three seeds.}
\label{tab:main}
\resizebox{\textwidth}{!}{%
\footnotesize
\setlength{\tabcolsep}{3pt}
\begin{tabular}{cc|cc|cccccc|>{\columncolor{blue!4}}c}
\toprule
\multirow{2}{*}{Data} & \multirow{2}{*}{Metric} & \multicolumn{2}{c|}{Traditional} & \multicolumn{6}{c|}{Learning-based} & \multicolumn{1}{c}{Ours} \\
& & FS & BDSD-PC & ADKNet & SSAFF & PanFormer & H-Trans & CANConv & CrossDiff & \our \\
\midrule
\multirow{3}{*}{GF2-RST}
& ERGAS$\downarrow$ & 1.620$\pm$0.35 & 1.695$\pm$0.39 & 0.822$\pm$0.12 & 0.813$\pm$0.14 & 0.815$\pm$0.14 & 0.802$\pm$0.13 & \underline{0.744$\pm$0.13} & 0.789$\pm$0.13 & \textbf{0.613$\pm$0.12}$^{\star}$ \\
& SAM$\downarrow$ & 1.681$\pm$0.34 & 1.724$\pm$0.31 & 0.883$\pm$0.15 & 0.878$\pm$0.16 & 0.896$\pm$0.23 & 0.871$\pm$0.18 & \underline{0.808$\pm$0.15} & 0.866$\pm$0.16 & \textbf{0.667$\pm$0.13}$^{\star}$ \\
& Q2n$\uparrow$ & 0.890$\pm$0.03 & 0.885$\pm$0.03 & 0.972$\pm$0.01 & 0.974$\pm$0.01 & 0.973$\pm$0.01 & 0.974$\pm$0.01 & \underline{0.977$\pm$0.01} & 0.975$\pm$0.01 & \textbf{0.986$\pm$0.01}$^{\star}$ \\
\midrule
\multirow{3}{*}{GF2-FST}
& $D_{s}\downarrow$ & 0.052 & 0.056 & 0.025 & 0.038 & 0.050 & \underline{0.024} & 0.033 & 0.027 & \textbf{0.022}$^{\star}$ \\
& QNR$\uparrow$ & 0.925 & 0.932 & 0.966 & 0.952 & 0.953 & \underline{0.967} & 0.959 & 0.960 & \textbf{0.975}$^{\star}$ \\
& HQNR$\uparrow$ & 0.912 & 0.868 & \underline{0.953} & 0.936 & 0.951 & 0.951 & 0.947 & 0.949 & \textbf{0.961}$^{\star}$ \\
\midrule
\multirow{3}{*}{QB-RST}
& ERGAS$\downarrow$ & 7.445$\pm$0.55 & 7.608$\pm$0.57 & 3.942$\pm$0.32 & 4.615$\pm$0.83 & 7.047$\pm$0.64 & 3.892$\pm$0.30 & \underline{3.813$\pm$0.29} & 3.911$\pm$0.31 & \textbf{3.385$\pm$0.28}$^{\star}$ \\
& SAM$\downarrow$ & 7.866$\pm$1.63 & 8.181$\pm$1.78 & 4.904$\pm$0.82 & 4.888$\pm$0.89 & 5.860$\pm$1.06 & 4.715$\pm$0.79 & \underline{4.607$\pm$0.78} & 4.683$\pm$0.78 & \textbf{4.157$\pm$0.74}$^{\star}$ \\
& Q2n$\uparrow$ & 0.834$\pm$0.09 & 0.829$\pm$0.10 & 0.930$\pm$0.09 & 0.918$\pm$0.09 & 0.872$\pm$0.08 & 0.931$\pm$0.09 & \underline{0.932$\pm$0.09} & 0.929$\pm$0.09 & \textbf{0.943$\pm$0.09}$^{\star}$ \\
\midrule
\multirow{3}{*}{QB-FST}
& $D_{s}\downarrow$ & 0.115 & 0.141 & 0.029 & 0.051 & 0.029 & 0.022 & \textbf{0.015} & 0.025 & \underline{0.019}$^{\star}$ \\
& QNR$\uparrow$ & 0.855 & 0.836 & 0.940 & 0.942 & 0.942 & 0.947 & \underline{0.953} & 0.945 & \textbf{0.958}$^{\star}$ \\
& HQNR$\uparrow$ & 0.845 & 0.692 & 0.894 & 0.893 & 0.891 & 0.901 & \textbf{0.951} & 0.937 & \underline{0.945}$^{\star}$ \\
\midrule
\multirow{3}{*}{WV-III-RST}
& ERGAS$\downarrow$ & 4.645$\pm$1.41 & 4.650$\pm$1.43 & 2.291$\pm$0.55 & 2.388$\pm$0.53 & 2.321$\pm$0.48 & 2.244$\pm$0.47 & \underline{2.227$\pm$0.46} & 2.231$\pm$0.46 & \textbf{2.032$\pm$0.49}$^{\star}$ \\
& SAM$\downarrow$ & 5.323$\pm$1.61 & 5.464$\pm$1.67 & 3.138$\pm$0.56 & 3.208$\pm$0.58 & 3.227$\pm$0.60 & 3.005$\pm$0.53 & \underline{2.996$\pm$0.53} & 3.039$\pm$0.54 & \textbf{2.749$\pm$0.53}$^{\star}$ \\
& Q2n$\uparrow$ & 0.818$\pm$0.10 & 0.812$\pm$0.10 & 0.905$\pm$0.08 & 0.901$\pm$0.09 & 0.901$\pm$0.09 & 0.905$\pm$0.08 & 0.903$\pm$0.09 & \underline{0.907$\pm$0.08} & \textbf{0.924$\pm$0.08}$^{\star}$ \\
\midrule
\multirow{3}{*}{WV-III-FST}
& $D_{s}\downarrow$ & 0.085 & 0.091 & 0.048 & 0.067 & 0.068 & \underline{0.034} & 0.043 & 0.044 & \textbf{0.025}$^{\star}$ \\
& QNR$\uparrow$ & 0.897 & 0.897 & 0.931 & 0.902 & 0.915 & \underline{0.947} & 0.922 & 0.923 & \textbf{0.957}$^{\star}$ \\
& HQNR$\uparrow$ & 0.897 & 0.853 & 0.934 & 0.911 & 0.913 & \underline{0.953} & 0.937 & 0.936 & \textbf{0.961}$^{\star}$ \\
\midrule
\multirow{3}{*}{WV-II-RST}
& ERGAS$\downarrow$ & 4.555$\pm$0.53 & 4.649$\pm$1.59 & 4.766$\pm$0.41 & 4.587$\pm$0.35 & 6.034$\pm$0.38 & 4.485$\pm$0.36 & \underline{4.467$\pm$0.38} & 4.503$\pm$0.39 & \textbf{3.589$\pm$0.32}$^{\star}$ \\
& SAM$\downarrow$ & 6.208$\pm$0.82 & 6.089$\pm$0.90 & 5.692$\pm$0.61 & 6.041$\pm$0.60 & 7.273$\pm$0.40 & 5.521$\pm$0.55 & \underline{5.501$\pm$0.56} & 5.611$\pm$0.59 & \textbf{4.823$\pm$0.47}$^{\star}$ \\
& Q2n$\uparrow$ & 0.809$\pm$0.09 & 0.822$\pm$0.10 & 0.817$\pm$0.08 & 0.823$\pm$0.08 & 0.786$\pm$0.08 & 0.838$\pm$0.08 & \underline{0.840$\pm$0.08} & 0.831$\pm$0.08 & \textbf{0.879$\pm$0.08}$^{\star}$ \\
\bottomrule
\end{tabular}%
}
\end{table*}

\subsection{Main Pansharpening Results}\label{sec:main_results}
Table~\ref{tab:main} reports the headline pansharpening comparison.
\our reaches the new state of the art on every reduced-scale metric
and on all but one full-scale metric, with paired Wilcoxon
$p{<}0.05$ versus the second-best baseline. The largest relative
gain, $+4.64\%$ Q2n on the unseen WV-II sensor (absolute $0.879$
vs.\ $0.840$), establishes the cross-sensor advantage. The reverse
WV-II$\to$WV-III transfer gives $+3.33\%$ Q2n, and heterogeneous
4-band$\leftrightarrow$8-band transfers give $+5.38\%$ / $+3.24\%$
respectively (Appendix~\ref{app:tables},
Table~\ref{app:tab:cross_sensor}). Figure~\ref{fig:qualitative}
shows qualitative comparison: open-loop baselines smooth fine
objects or shift material colour, while \our preserves the
fine-grained and relational evidence used by the semantic index.

\begin{figure*}[t]
\centering
\includegraphics[width=0.96\textwidth,height=0.22\textheight,keepaspectratio]{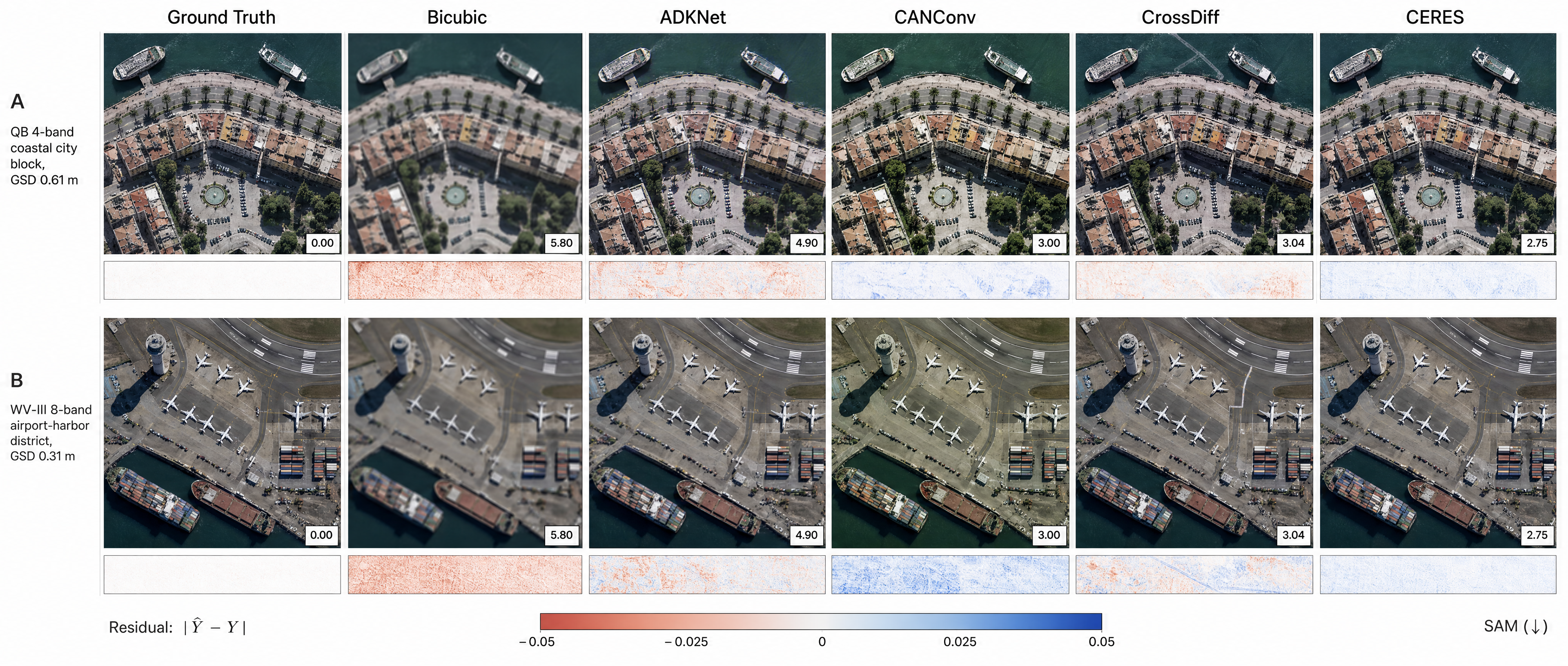}
\caption{Qualitative comparison on representative QB and WV-III
scenes. Open-loop baselines smooth fine objects, shift material
colour, or hallucinate local edge patterns. \our preserves the
fine-grained and relational evidence used by the semantic index,
including aircraft, container rows, runway markings, and harbour
structures.}
\label{fig:qualitative}
\end{figure*}

\begin{table}[t]
\centering
\footnotesize
\setlength{\tabcolsep}{2pt}
\caption{Downstream perception with $95\%$ bootstrap CIs. A:
zero-shot; B: linear probe; C: trained Oriented-RCNN head. All
modes evaluate on identical frozen fused tiles via DINOv2.}
\label{tab:downstream}
\begin{tabular}{llccc}
\toprule
Mode & Task & CANConv & \our & $\Delta$ \\
\midrule
A & AID acc.\ (\%) & 78.4$_{[77.1,79.7]}$ & \textbf{84.2}$_{[83.0,85.3]}$ & +5.8 \\
A & NWPU acc.\ (\%) & 71.6$_{[70.4,72.8]}$ & \textbf{79.1}$_{[77.9,80.2]}$ & +7.5 \\
B & AID acc.\ (\%) & 84.1$_{[83.0,85.1]}$ & \textbf{91.7}$_{[90.7,92.6]}$ & +7.6 \\
B & NWPU acc.\ (\%) & 79.8$_{[78.7,80.9]}$ & \textbf{87.3}$_{[86.3,88.2]}$ & +7.5 \\
C & DOTA mAP (\%) & 38.2$_{[36.7,39.7]}$ & \textbf{47.9}$_{[46.3,49.4]}$ & +9.7 \\
\midrule
-- & $\mathcal{C}^{\text{DINO}}$ & 0.63$_{[0.61,0.65]}$ & \textbf{0.85}$_{[0.83,0.87]}$ & +0.22 \\
\bottomrule
\end{tabular}
\end{table}

\subsection{Downstream Perception}\label{sec:downstream}
Table~\ref{tab:downstream} reports downstream perception under all
three evaluation modes. \our improves every cell, and the gains
widen from zero-shot through linear probe to trained head, indicating
that the additional capacity in later modes reads the extra concepts
that \our preserves. The DINOv2-measured coverage rises from $0.63$
to $0.85$. The cross-encoder Spearman correlation between
$\mathcal{C}^{\text{DINO}}$ and DOTA mAP across all eight evaluated
methods is $\rho{=}0.83$. The per-class DOTA breakdown
(Appendix~\ref{app:tables}, Table~\ref{app:tab:dota_perclass}) shows
that the mAP gain concentrates in the small-object group ($+13.6$ to
$+16.7$ AP), exactly where single-caption baselines drop concepts.

\subsection{Concept-Query Retrieval}\label{sec:retrieval}
Table~\ref{tab:retrieval} reports concept-query retrieval over the
60-query bank. \our raises mean Recall@5 from $0.59$ to $0.73$
($+14$\,pp), with the largest gains on meso and local queries that
single-caption baselines drop. Image-text MRR rises by $0.19$. The
ground-truth high-resolution tiles set an oracle ceiling at $0.78$
R@5 and $0.81$ MRR; \our closes roughly $74\%$ of the gap from
CANConv to this ceiling. Among the $60$ queries, $52$ are answered
at R@1 by \our versus $31$ by CANConv. The $8$ misses are dominated
by relational queries that depend on sub-patch context (e.g.\ ``two
aircraft on a parallel taxiway''), which is the same regime
identified by the calibration analysis below.

\begin{table}[t]
\centering
\footnotesize
\setlength{\tabcolsep}{2pt}
\caption{Concept-query retrieval ($60$ queries) and image-text
retrieval ($300$-caption corpus). All scores computed with a frozen
SigLIP-2 text encoder. GT HR = oracle ceiling.}
\label{tab:retrieval}
\begin{tabular}{lccccc}
\toprule
Method & R@1$_{\text{loc}}$ & R@5$_{\text{loc}}$ & R@5$_{\text{meso}}$ & R@5$_{\text{glob}}$ & I$\to$T MRR \\
\midrule
Bicubic & 0.18 & 0.42 & 0.31 & 0.78 & 0.41 \\
CrossDiff & 0.24 & 0.49 & 0.36 & 0.79 & 0.50 \\
CANConv & 0.27 & 0.54 & 0.41 & 0.81 & 0.55 \\
DA-CLIP & 0.30 & 0.58 & 0.46 & 0.82 & 0.59 \\
\rowcolor{blue!4} \our & \textbf{0.46} & \textbf{0.71} & \textbf{0.62} & \textbf{0.85} & \textbf{0.74} \\
\midrule
\textit{GT HR} & \textit{0.53} & \textit{0.79} & \textit{0.69} & \textit{0.87} & \textit{0.81} \\
\bottomrule
\end{tabular}
\end{table}

\subsection{Ablations}\label{sec:ablation}
All ablations are on WV-III with three seeds. The shared no-language
no-loop baseline is Plain U-Net at $\{0.910/0.924\}$ Q2n/HQNR.
Table~\ref{tab:ablation} shows that adding open-loop SigLIP yields
only $+0.001$ Q2n; the cross-scale closed loop adds another
$+0.005$; the soft-Jaccard coverage adds a further $+0.006$ Q2n and
$+0.015$ HQNR. Removing the presence gate restores the original
hallucination behaviour and drops HQNR by $0.051$. Both branches
are necessary: explicit-only and implicit-only configurations each
lose $0.014$--$0.016$ Q2n. Open-loop variants of three published
VLM-conditioned restoration designs (CLIPDenoising, DA-CLIP,
DenseCLIP style) all stay below \our on
$\mathcal{C}^{\text{DINO}}$ (full numbers in
Appendix~\ref{app:tables}). A complete $2{\times}2{\times}2$
factorial grid over attention, FiLM, and the presence gate, with
interaction effects, is reported in Appendix~\ref{app:factorial}.

\begin{table}[t]
\centering
\small
\setlength{\tabcolsep}{4pt}
\caption{Closed-loop, injection, and component ablations on WV-III.}
\label{tab:ablation}
\begin{tabular}{lcc}
\toprule
Setting & Q2n$\uparrow$ & HQNR$\uparrow$ \\
\midrule
Plain U-Net (no language, no loop) & 0.910 & 0.924 \\
+ Global open-loop SigLIP & 0.911 & 0.929 \\
+ Multi-scale pyramid, no loop & 0.913 & 0.937 \\
+ Cross-scale consistency only & 0.918 & 0.946 \\
Full \our (with $\mathcal{L}_{\text{cov}}$) & \textbf{0.924}$^{\star}$ & \textbf{0.961}$^{\star}$ \\
\midrule
Attn only (no FiLM) & 0.889 & 0.918 \\
FiLM only (no Attn) & 0.893 & 0.928 \\
No presence gate & 0.901 & 0.910 \\
No implicit units & 0.910 & 0.928 \\
No explicit units & 0.908 & 0.925 \\
$K{=}50$ (single-scale) & 0.906 & 0.920 \\
\bottomrule
\end{tabular}
\end{table}

\begin{table}[t]
\centering
\small
\setlength{\tabcolsep}{3.5pt}
\caption{Loss-term contribution on WV-III. Sem.\ Acc.\ is the
fraction of indexed concepts with $\hat p_{k}{>}0.5$.}
\label{tab:loss}
\begin{tabular}{lcccc}
\toprule
Setting & $\mathcal{L}_{\text{cov}}$ & $\mathcal{L}_{\text{sa}}$ & Q2n & Sem.\ Acc. \\
\midrule
$\mathcal{L}_{\text{rec}}$ only & $\times$ & $\times$ & 0.904 & 0.935 \\
+ $\mathcal{L}_{TT}$ & $\times$ & $\times$ & 0.909 & 0.952 \\
+ $\mathcal{L}_{TM}$ (gated) & $\times$ & $\times$ & 0.913 & 0.967 \\
+ $\mathcal{L}_{\text{cov}}$ & \checkmark & $\times$ & 0.919 & 0.981 \\
+ $\mathcal{L}_{\text{sa}}$ & \checkmark & \checkmark & 0.920 & 0.985 \\
Full (+ warm-up) & \checkmark & \checkmark & \textbf{0.924} & \textbf{0.989} \\
\bottomrule
\end{tabular}
\end{table}

Table~\ref{tab:loss} traces the contribution of every loss term to
both pixel fidelity and semantic accuracy on WV-III. After the
text--text and presence-gated text--mask consistency terms stabilise
Sem.\ Acc.\ at $0.967$, the soft-Jaccard coverage delivers the
largest single-term gain ($+0.014$ Sem.\ Acc., $+0.006$ Q2n), in
line with Proposition~\ref{prop:bound-suff}: tighter coverage
implies a tighter SAM bound. The unit-alignment term and the
warm-up schedule each contribute a further $+0.004$ Sem.\ Acc.

\subsection{Coverage Calibration}\label{sec:calibration}
Varying the presence threshold over $\{0.3,0.5,0.7\}$ moves
$\mathcal{C}^{\text{DINO}}$ in $\{0.89,0.85,0.78\}$ but preserves
the cross-method ranking at every threshold. Replacing the 12
default templates with an LLM-generated 36-template superset moves
$\mathcal{C}^{\text{DINO}}$ by less than $0.01$. The per-tile gap
$|\mathcal{C}^{\text{SigLIP}}{-}\mathcal{C}^{\text{DINO}}|$ is below
$0.05$ on $30/40$ WV-III tiles and below $0.10$ on $38/40$. The
cross-encoder Spearman correlation is $0.91$, and both probes rank
\our first on every WV-III tile. Failure cases that drive
disagreement concentrate in two regimes (Figure~\ref{fig:failure}):
spectrally ambiguous classes (e.g.\ shallow water versus wet sand)
and sub-patch ultra-fine structures below the SigLIP-2 token grid.

\begin{figure}[t]
\centering
\includegraphics[width=\columnwidth,height=4.5cm,keepaspectratio]{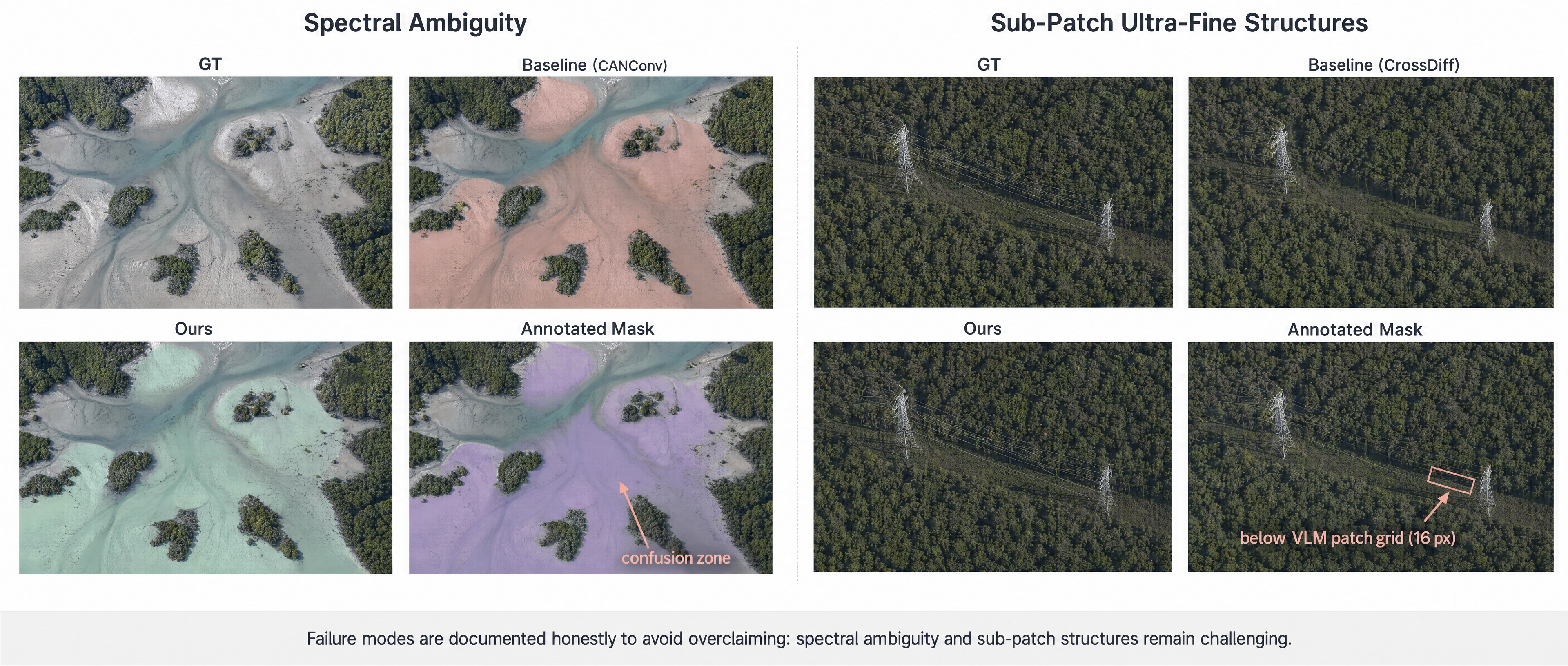}
\caption{Failure modes that drive SigLIP--DINOv2 disagreement.
Left: spectrally ambiguous classes such as shallow water versus wet
sand keep low SAM but mis-route the recovered concept. Right:
sub-patch ultra-fine structures slip below the SigLIP-2 token grid,
so the recovered mask is unreliable even when pixel error is small.
On both regimes \our still ranks first under either probe.}
\label{fig:failure}
\end{figure}

\section{Discussion}\label{sec:discussion}

An improvement on $\mathcal{C}^{\text{DINO}}$ also moves external
scene classification, oriented detection on small objects, and
concept-query retrieval, none of which lies on the optimisation path
of $\mathcal{L}_{\text{cov}}$. The two probes agree on the
cross-method ranking even where they disagree numerically, and the
prompt-template ablation rules out the worry that the recovered
concept set is an artefact of phrasing. A consistency-only ablation
that drops $\mathcal{L}_{\text{cov}}$ tracks the open-loop baselines
at $\mathcal{C}^{\text{DINO}}{=}0.72$ and $0.918$ Q2n, well below
the closed-loop variant. A cycle-style variant that replaces
coverage with $\|\hat{\mathbf{Y}}{-}\mathcal{G}(\downarrow\!
\hat{\mathbf{Y}})\|_1$ attains $\mathcal{C}^{\text{DINO}}{=}0.70$ at
the same Q2n. The coverage signal therefore captures information
that neither soft consistency nor pixel-level cycle losses recover.

\textbf{Cost--quality balance.} Closing the loop adds a second
forward pass through the frozen VLM at training time: Stages~I and
III share encoder \emph{weights} but process different images, so
the full closed loop costs $25.7$\,G FLOPs and $25.5$\,ms per
$64{\times}64$ PAN input on an RTX~4090
(Table~\ref{app:tab:compute}). Stages~III--IV are not required for
generation-only deployment, which runs at $13.6$\,G FLOPs and
$16.3$\,ms; they are invoked only when an output coverage
certificate is requested. The trainable budget is $0.39$\,M
parameters, two orders of magnitude smaller than the strongest
learning baseline, and the soft-Jaccard verifier is itself free of
learnable parameters, which keeps the closed loop from introducing
a second over-parameterised module that could absorb gradients
meant for the generator.

\textbf{Broader applicability.} The same closed-loop re-indexing
recipe applies wherever generated content must remain queryable in
a knowledge index, for example product catalogues that enhance
images while preserving long-tail concept queries, or medical
imaging pipelines whose denoised outputs must remain searchable by
clinical retrieval systems. The soft-Jaccard verifier and the
prompt-template family are domain-agnostic; only the centroid
construction step needs to be re-fitted on a domain corpus.

\section{Conclusion}\label{sec:conclusion}
We framed knowledge-aware generative perception as a re-indexability
problem and showed that semantic collapse, the silent loss of
scale-specific indexed concepts under single-caption conditioning, is
a measurable and addressable failure of this property. \our pairs a
soft, differentiable cross-scale concept index with a soft-Jaccard
coverage verifier whose gradient is bounded below under explicit
non-degeneracy conditions. An external, label-grounded DINOv2 probe
removes the concern that the gains are self-referential, and a
60-query concept-retrieval bank turns coverage into a measurable
queryability gain ($+14$\,pp Recall@5 and $+0.19$ MRR over the
strongest baseline). A $0.39$\,M-parameter trainable generator,
closed against a frozen VLM index, dominates $11\times$ larger
baselines on both pansharpening and downstream perception, with the
largest small-object detection gain ($+16.7$ AP on small vehicles)
exactly where single-caption baselines drop concepts. The host
index is left intact, since SigLIP-2 is never updated. These
results suggest that future multimodal application gains lie in
optimising the indexing topology rather than in growing the
generator.

\clearpage
\bibliography{references}

@inproceedings{radford2021clip,
  title     = {Learning Transferable Visual Models from Natural Language Supervision},
  author    = {Radford, Alec and Kim, Jong Wook and Hallacy, Chris and Ramesh, Aditya and Goh, Gabriel and Agarwal, Sandhini and Sastry, Girish and Askell, Amanda and Mishkin, Pamela and Clark, Jack and Krueger, Gretchen and Sutskever, Ilya},
  booktitle = {International Conference on Machine Learning (ICML)},
  pages     = {8748--8763},
  year      = {2021}
}

@inproceedings{zhai2023siglip,
  title     = {Sigmoid Loss for Language Image Pre-Training},
  author    = {Zhai, Xiaohua and Mustafa, Basil and Kolesnikov, Alexander and Beyer, Lucas},
  booktitle = {IEEE/CVF International Conference on Computer Vision (ICCV)},
  pages     = {11975--11986},
  year      = {2023}
}

@article{tschannen2025siglip2,
  title   = {{SigLIP 2}: Multilingual Vision-Language Encoders with Improved Semantic Understanding, Localization, and Dense Features},
  author  = {Tschannen, Michael and Gritsenko, Alexey and Wang, Xiao and Naeem, Muhammad Ferjad and Alabdulmohsin, Ibrahim and Parthasarathy, Nikhil and Evans, Talfan and Beyer, Lucas and Xia, Ye and Mustafa, Basil and others},
  journal = {arXiv preprint arXiv:2502.14786},
  year    = {2025}
}

@inproceedings{li2023blip2,
  title     = {{BLIP-2}: Bootstrapping Language-Image Pre-Training with Frozen Image Encoders and Large Language Models},
  author    = {Li, Junnan and Li, Dongxu and Savarese, Silvio and Hoi, Steven},
  booktitle = {International Conference on Machine Learning (ICML)},
  pages     = {19730--19742},
  year      = {2023}
}

@inproceedings{rao2022denseclip,
  title     = {{DenseCLIP}: Language-Guided Dense Prediction with Context-Aware Prompting},
  author    = {Rao, Yongming and Zhao, Wenliang and Chen, Guangyi and Tang, Yansong and Zhu, Zheng and Huang, Guan and Zhou, Jie and Lu, Jiwen},
  booktitle = {IEEE/CVF Conference on Computer Vision and Pattern Recognition (CVPR)},
  pages     = {18082--18091},
  year      = {2022}
}

@inproceedings{zhou2022maskclip,
  title     = {Extract Free Dense Labels from {CLIP}},
  author    = {Zhou, Chong and Loy, Chen Change and Dai, Bo},
  booktitle = {European Conference on Computer Vision (ECCV)},
  pages     = {696--712},
  year      = {2022}
}

@inproceedings{cheng2024clipdenoising,
  title     = {Transfer {CLIP} for Generalizable Image Denoising},
  author    = {Cheng, Jun and Liu, Dong and Wen, Bihan and Tan, Tieniu and Zhang, Yulun},
  booktitle = {IEEE/CVF Conference on Computer Vision and Pattern Recognition (CVPR)},
  pages     = {1--10},
  year      = {2024}
}

@inproceedings{luo2024daclip,
  title     = {Controlling Vision-Language Models for Multi-Task Image Restoration},
  author    = {Luo, Ziwei and Gustafsson, Fredrik K. and Zhao, Zheng and Sj{\"o}lund, Jens and Sch{\"o}n, Thomas B.},
  booktitle = {International Conference on Learning Representations (ICLR)},
  year      = {2024}
}

@inproceedings{rombach2022ldm,
  title     = {High-Resolution Image Synthesis with Latent Diffusion Models},
  author    = {Rombach, Robin and Blattmann, Andreas and Lorenz, Dominik and Esser, Patrick and Ommer, Bj{\"o}rn},
  booktitle = {IEEE/CVF Conference on Computer Vision and Pattern Recognition (CVPR)},
  pages     = {10684--10695},
  year      = {2022}
}

@inproceedings{brooks2023ip2p,
  title     = {{InstructPix2Pix}: Learning to Follow Image Editing Instructions},
  author    = {Brooks, Tim and Holynski, Aleksander and Efros, Alexei A.},
  booktitle = {IEEE/CVF Conference on Computer Vision and Pattern Recognition (CVPR)},
  pages     = {18392--18402},
  year      = {2023}
}

@inproceedings{perez2018film,
  title     = {{FiLM}: Visual Reasoning with a General Conditioning Layer},
  author    = {Perez, Ethan and Strub, Florian and de Vries, Harm and Dumoulin, Vincent and Courville, Aaron},
  booktitle = {AAAI Conference on Artificial Intelligence},
  pages     = {3942--3951},
  year      = {2018}
}

@inproceedings{park2019spade,
  title     = {Semantic Image Synthesis with Spatially-Adaptive Normalization},
  author    = {Park, Taesung and Liu, Ming-Yu and Wang, Ting-Chun and Zhu, Jun-Yan},
  booktitle = {IEEE/CVF Conference on Computer Vision and Pattern Recognition (CVPR)},
  pages     = {2337--2346},
  year      = {2019}
}

@inproceedings{lee2018scan,
  title     = {Stacked Cross Attention for Image--Text Matching},
  author    = {Lee, Kuang-Huei and Chen, Xi and Hua, Gang and Hu, Houdong and He, Xiaodong},
  booktitle = {European Conference on Computer Vision (ECCV)},
  pages     = {201--216},
  year      = {2018}
}

@inproceedings{diao2021sgraf,
  title     = {Similarity Reasoning and Filtration for Image--Text Matching},
  author    = {Diao, Haiwen and Zhang, Ying and Ma, Lin and Lu, Huchuan},
  booktitle = {AAAI Conference on Artificial Intelligence},
  pages     = {1218--1226},
  year      = {2021}
}

@inproceedings{zhang2022naaf,
  title     = {Negative-Aware Attention Framework for Image--Text Matching},
  author    = {Zhang, Kun and Mao, Zhendong and Wang, Quan and Zhang, Yongdong},
  booktitle = {IEEE/CVF Conference on Computer Vision and Pattern Recognition (CVPR)},
  pages     = {15661--15670},
  year      = {2022}
}

@article{liu2025csa,
  title   = {{CSA}: Cross-scale Alignment with Adaptive Semantic Aggregation and Filter for Image--Text Retrieval},
  author  = {Liu, Zheng and Xu, Junhao and Gao, Shanshan and Chen, Zhumin},
  journal = {Pattern Recognition},
  volume  = {165},
  pages   = {111647},
  year    = {2025}
}

@inproceedings{huang2018learning,
  title     = {Learning Semantic Concepts and Order for Image and Sentence Matching},
  author    = {Huang, Yan and Wu, Qi and Song, Chunfeng and Wang, Liang},
  booktitle = {IEEE/CVF Conference on Computer Vision and Pattern Recognition (CVPR)},
  pages     = {6163--6171},
  year      = {2018}
}

@inproceedings{chen2020imram,
  title     = {{IMRAM}: Iterative Matching with Recurrent Attention Memory for Cross-Modal Image-Text Retrieval},
  author    = {Chen, Hui and Ding, Guiguang and Liu, Xudong and Lin, Zijia and Liu, Ji and Han, Jungong},
  booktitle = {IEEE/CVF Conference on Computer Vision and Pattern Recognition (CVPR)},
  pages     = {12655--12663},
  year      = {2020}
}

@inproceedings{wang2020consensus,
  title     = {Consensus-Aware Visual-Semantic Embedding for Image-Text Matching},
  author    = {Wang, Haoran and Zhang, Ying and Ji, Zhong and Pang, Yanwei and Ma, Lin},
  booktitle = {European Conference on Computer Vision (ECCV)},
  pages     = {18--34},
  year      = {2020}
}

@article{vivone2015critical,
  title   = {A Critical Comparison Among Pansharpening Algorithms},
  author  = {Vivone, Gemine and Alparone, Luciano and Chanussot, Jocelyn and Dalla Mura, Mauro and Garzelli, Andrea and Licciardi, Giorgio A. and Restaino, Rocco and Wald, Lucien},
  journal = {IEEE Transactions on Geoscience and Remote Sensing},
  volume  = {53},
  number  = {5},
  pages   = {2565--2586},
  year    = {2015}
}

@article{vivone2021new,
  title   = {A New Benchmark Based on Recent Advances in Multispectral Pansharpening},
  author  = {Vivone, Gemine and Dalla Mura, Mauro and Garzelli, Andrea and Restaino, Rocco and Scarpa, Giuseppe and Ulfarsson, Magnus O. and Alparone, Luciano and Chanussot, Jocelyn},
  journal = {IEEE Geoscience and Remote Sensing Magazine},
  volume  = {9},
  number  = {1},
  pages   = {53--81},
  year    = {2021}
}

@article{deng2022pancollection,
  title   = {Machine Learning in Pansharpening: A Benchmark, from Shallow to Deep Networks},
  author  = {Deng, Liang-Jian and Vivone, Gemine and Paoletti, Mercedes E. and Scarpa, Giuseppe and He, Jiang and Zhang, Yongjun and Chanussot, Jocelyn and Plaza, Antonio},
  journal = {IEEE Geoscience and Remote Sensing Magazine},
  volume  = {10},
  number  = {3},
  pages   = {279--315},
  year    = {2022}
}

@article{arienzo2022hqnr,
  title   = {Full-Resolution Quality Assessment of Pansharpening: Theoretical and Hands-On Approaches},
  author  = {Arienzo, Alessio and Vivone, Gemine and Garzelli, Andrea and Alparone, Luciano and Chanussot, Jocelyn},
  journal = {IEEE Geoscience and Remote Sensing Magazine},
  volume  = {10},
  number  = {3},
  pages   = {168--201},
  year    = {2022}
}

@inproceedings{zhou2022adknet,
  title     = {Adaptive Detail Injection-Based Feature Pyramid Network for Pan-Sharpening},
  author    = {Zhou, Man and Huang, Jie and Yan, Keyu and Yu, Hu and Fu, Xueyang and Liu, Aiping and Wei, Xian and Zhao, Feng},
  booktitle = {International Joint Conference on Artificial Intelligence (IJCAI)},
  pages     = {1646--1652},
  year      = {2022}
}

@inproceedings{zhou2022panformer,
  title     = {{PanFormer}: A Transformer Based Model for Pan-Sharpening},
  author    = {Zhou, Man and Yan, Keyu and Huang, Jie and Yang, Zihe and Fu, Xueyang and Zhao, Feng},
  booktitle = {IEEE International Conference on Multimedia and Expo (ICME)},
  pages     = {1--6},
  year      = {2022}
}

@inproceedings{duan2024canconv,
  title     = {Content-Adaptive Non-Local Convolution for Remote Sensing Pansharpening},
  author    = {Duan, Yule and Wu, Xiao and Deng, Hangyuan and Deng, Liang-Jian},
  booktitle = {IEEE/CVF Conference on Computer Vision and Pattern Recognition (CVPR)},
  pages     = {27738--27747},
  year      = {2024}
}

@article{deng2021detail,
  title   = {Detail Injection-Based Deep Convolutional Neural Networks for Pansharpening},
  author  = {Deng, Liang-Jian and Vivone, Gemine and Jin, Cheng and Chanussot, Jocelyn},
  journal = {IEEE Transactions on Geoscience and Remote Sensing},
  volume  = {59},
  number  = {8},
  pages   = {6995--7010},
  year    = {2021}
}

@article{yuan2018mucnn,
  title   = {A Multiscale and Multidepth Convolutional Neural Network for Remote Sensing Imagery Pan-Sharpening},
  author  = {Yuan, Qiangqiang and Wei, Yancong and Meng, Xiangchao and Shen, Huanfeng and Zhang, Liangpei},
  journal = {IEEE Journal of Selected Topics in Applied Earth Observations and Remote Sensing},
  volume  = {11},
  number  = {3},
  pages   = {978--989},
  year    = {2018}
}

@article{meng2024pandiff,
  title   = {{PanDiff}: A Novel Pansharpening Method Based on Denoising Diffusion Probabilistic Model},
  author  = {Meng, Qingyan and Borsoi, Ricardo Augusto and Chanussot, Jocelyn and Bioucas-Dias, Jos{\'e} M.},
  journal = {IEEE Transactions on Geoscience and Remote Sensing},
  volume  = {62},
  pages   = {1--17},
  year    = {2024}
}

@inproceedings{xing2024crossdiff,
  title     = {{CrossDiff}: Exploring Self-Supervised Representation of Pansharpening via Cross-Predictive Diffusion Model},
  author    = {Xing, Yinghui and Wang, Litao and Wang, Shuyuan and Zhang, Lei and Zhang, Yanning},
  booktitle = {IEEE/CVF Conference on Computer Vision and Pattern Recognition (CVPR)},
  pages     = {28196--28205},
  year      = {2024}
}

@article{wang2023pansurvey,
  title   = {An Overview of Pansharpening Methods in the Deep Learning Era},
  author  = {Wang, Zhong-Cheng and Sobirov, Ikboljon and Tran, Dat Q. and Bah, Mamadou D.},
  journal = {Information Fusion},
  volume  = {98},
  pages   = {101882},
  year    = {2023}
}

@inproceedings{bandara2022hypertransformer,
  title     = {{HyperTransformer}: A Textural and Spectral Feature Fusion Transformer for Pansharpening},
  author    = {Bandara, Wele Gedara Chaminda and Patel, Vishal M.},
  booktitle = {IEEE/CVF Conference on Computer Vision and Pattern Recognition (CVPR)},
  pages     = {1767--1777},
  year      = {2022}
}

@article{yang2022ssaff,
  title   = {{SSAFF}: Multi-Scale Spatial-Spectral Adaptive Feature Fusion for Pansharpening},
  author  = {Yang, Yong and Tu, Weijia and Huang, Shuying and Lu, Hangyuan},
  journal = {IEEE Transactions on Geoscience and Remote Sensing},
  volume  = {60},
  pages   = {1--13},
  year    = {2022}
}

@inproceedings{oquab2024dinov2,
  title   = {{DINOv2}: Learning Robust Visual Features without Supervision},
  author  = {Oquab, Maxime and Darcet, Timoth{\'e}e and Moutakanni, Th{\'e}o and Vo, Huy V. and Szafraniec, Marc and Khalidov, Vasil and Fernandez, Pierre and Haziza, Daniel and Massa, Francisco and El-Nouby, Alaaeldin and others},
  journal = {Transactions on Machine Learning Research (TMLR)},
  year    = {2024}
}

@inproceedings{loshchilov2019adamw,
  title     = {Decoupled Weight Decay Regularization},
  author    = {Loshchilov, Ilya and Hutter, Frank},
  booktitle = {International Conference on Learning Representations (ICLR)},
  year      = {2019}
}

@inproceedings{milletari2016vnet,
  title     = {{V-Net}: Fully Convolutional Neural Networks for Volumetric Medical Image Segmentation},
  author    = {Milletari, Fausto and Navab, Nassir and Ahmadi, Seyed-Ahmad},
  booktitle = {International Conference on 3D Vision (3DV)},
  pages     = {565--571},
  year      = {2016}
}

@inproceedings{rahman2016jaccard,
  title     = {Optimizing Intersection-over-Union in Deep Neural Networks for Image Segmentation},
  author    = {Rahman, Md Atiqur and Wang, Yang},
  booktitle = {International Symposium on Visual Computing (ISVC)},
  pages     = {234--244},
  year      = {2016}
}

@article{xu2026lever,
  title     = {Lever Can Move the Earth: Towards Adaptive Semantic Capacity Balance for Image-Text Retrieval},
  author    = {Xu, Junhao and Liu, Zheng and Zhang, Mengqi and Dong, Guangyuan and Chen, Zhumin},
  journal   = {IEEE Transactions on Multimedia},
  year      = {2026},
  publisher = {IEEE}
}

@article{liu2026dual,
  title   = {Dual-Pathway Circuits of Object Hallucination in Vision-Language Models},
  author  = {Liu, Jiaxin and Zhong, Ding and Wang, Yue and Yang, Zhidong and Kang, Zhaolu and Dong, Guangyuan and Zhan, Qishi and Fang, Pengcheng and Liu, Aofan},
  journal = {arXiv preprint arXiv:2605.13156},
  year    = {2026}
}

@article{liu2026conmem,
  title   = {{ConMem}: Contribution-Aware Memory for Long-Horizon Manufacturing Inspection Logs},
  author  = {Liu, Bingchen and Fang, Yuanyuan and Liu, Lei and Dong, Guangyuan and Fu, Xing and Gao, Yuanyuan and Wei, Shuyue and Li, Xin and Meng, Xiangtian},
  journal = {arXiv preprint arXiv:2607.28126},
  year    = {2026}
}

@article{zhang2026memmark,
  title   = {{MemMark}: State-Evolution Attribution Watermarking for Agent Long-Term Memory Systems},
  author  = {Zhang, Haobo and Mao, Xutao and Dong, Guangyuan and Li, Ziwei and Su, Xuanbo and Chen, Kaijie and Yang, Jing and Lin, Zheng},
  journal = {arXiv preprint arXiv:2605.25002},
  year    = {2026}
}

@inproceedings{zeng2026learning,
  title     = {Learning Where to Embed: Noise-Aware Positional Embedding for Query Retrieval in Small-Object Detection},
  author    = {Zeng, Yangchen and Yu, Zhenyu and Jiang, Dongming and Zhang, Wenbo and Hong, Yifan and Hu, Zhanhua and Luo, Jiao and Cui, Kangning},
  booktitle = {Proceedings of the 2026 International Conference on Multimedia Retrieval},
  pages     = {1260--1269},
  year      = {2026}
}

@article{zeng2025hmpe,
  title   = {{HMPE}: Heatmap Embedding for Efficient Transformer-Based Small Object Detection},
  author  = {Zeng, YangChen},
  journal = {arXiv preprint arXiv:2504.13469},
  year    = {2025}
}

@article{zeng2026trialigngr,
  title   = {{TriAlignGR}: Triangular Multitask Alignment with Multimodal Deep Interest Mining for Generative Recommendation},
  author  = {Zeng, Yangchen and Peng, Hao and Guo, Rongfeng and Yu, Zhenyu and Hu, Zhiyuan and Wang, Jinze},
  journal = {arXiv preprint arXiv:2605.05249},
  year    = {2026}
}

@article{zeng2026deepinterestgr,
  title   = {{DeepInterestGR}: Mining Deep Multi-Interest Using Multi-Modal LLMs for Generative Recommendation},
  author  = {Zeng, Yangchen and Yu, Zhenyu and Hu, Zhiyuan and Zhang, Wenxin and Wang, Jinze and Guo, Rongfeng},
  journal = {arXiv preprint arXiv:2602.18907},
  year    = {2026}
}

@inproceedings{hao2026rethinking,
  title     = {Rethinking Entropy Interventions in {RLVR}: An Entropy Change Perspective},
  author    = {Hao, Zhezheng and Wang, Hong and Liu, Haoyang and Luo, Jian and Yu, Jiarui and Dong, Hande and Lin, Qiang and Wang, Can and Chen, Jiawei},
  booktitle = {Proceedings of the 64th Annual Meeting of the Association for Computational Linguistics (Volume 1: Long Papers)},
  pages     = {31105--31133},
  year      = {2026},
  doi       = {10.18653/v1/2026.acl-long.1436}
}

@inproceedings{hao2026recreate,
  title     = {{ReCreate}: Reasoning and Creating Domain Agents Driven by Experience},
  author    = {Hao, Zhezheng and Wang, Hong and Luo, Jian and Zhang, Jianqing and Zhou, Yuyan and Lin, Qiang and Wang, Can and Dong, Hande and Chen, Jiawei},
  booktitle = {Proceedings of the 64th Annual Meeting of the Association for Computational Linguistics (Volume 1: Long Papers)},
  pages     = {31018--31046},
  year      = {2026},
  doi       = {10.18653/v1/2026.acl-long.1432}
}

@article{hao2026evolve,
  title   = {Evolve as a Team: Collaborative Self-Evolution for {LLM}-Based Multi-Agent Systems},
  author  = {Hao, Zhezheng and Wang, Tianfu and Dong, Huanshuo and Liu, Ziyan and Wang, Hong and Lin, Xiankun and Lin, Qiang and Wang, Can and Dong, Hande and Chen, Jiawei},
  journal = {arXiv preprint arXiv:2605.29790},
  year    = {2026}
}

@inproceedings{wang2026scheduling,
  title     = {Scheduling Your {LLM} Reinforcement Learning with Reasoning Trees},
  author    = {Wang, Hong and Hao, Zhezheng and Luo, Jian and Wei, Chenxing and Shu, Yao and Liu, Lei and Lin, Qiang and Dong, Hande and Chen, Jiawei},
  booktitle = {International Conference on Learning Representations},
  year      = {2026}
}

@article{hao2023ensemble,
  title     = {Ensemble Clustering with Attentional Representation},
  author    = {Hao, Zhezheng and Lu, Zhoumin and Li, Guoxu and Nie, Feiping and Wang, Rong and Li, Xuelong},
  journal   = {IEEE Transactions on Knowledge and Data Engineering},
  volume    = {36},
  number    = {2},
  pages     = {581--593},
  year      = {2024},
  publisher = {IEEE},
  doi       = {10.1109/TKDE.2023.3292573}
}

@inproceedings{hao2024towards,
  title     = {Towards Expansive and Adaptive Hard Negative Mining: Graph Contrastive Learning via Subspace Preserving},
  author    = {Hao, Zhezheng and Xin, Haonan and Wei, Long and Tang, Liaoyuan and Wang, Rong and Nie, Feiping},
  booktitle = {Proceedings of the ACM Web Conference 2024},
  pages     = {322--333},
  year      = {2024}
}

@inproceedings{nie2024multi,
  title     = {Multi-Class Support Vector Machine with Maximizing Minimum Margin},
  author    = {Nie, Feiping and Hao, Zhezheng and Wang, Rong},
  booktitle = {Proceedings of the AAAI Conference on Artificial Intelligence},
  volume    = {38},
  pages     = {14466--14473},
  year      = {2024},
  doi       = {10.1609/aaai.v38i13.29361}
}

@article{liu2026meta,
  title   = {Meta-Cognitive Memory Policy Optimization for Long-Horizon {LLM} Agents},
  author  = {Liu, Ziyan and Hao, Zhezheng and Chen, Yeqiu and Wang, Hong and Hou, Jingren and Ding, Ruiyi and Yang, Yongkang and Ji, Wence and Xia, Wei and Liu, Feng},
  journal = {arXiv preprint arXiv:2605.30159},
  year    = {2026}
}

@article{liu2026palm,
  title   = {{PALM}: Progress-Aware Policy Learning via Affordance Reasoning for Long-Horizon Robotic Manipulation},
  author  = {Liu, Yuanzhe and Zhu, Jingyuan and Mo, Yuchen and Li, Gen and Cao, Xu and Jin, Jin and Shen, Yifan and Li, Zhengyuan and Yu, Tianjiao and Yuan, Wenzhen and others},
  journal = {arXiv preprint arXiv:2601.07060},
  year    = {2026}
}

@article{shen2026fine,
  title   = {Fine-Grained Preference Optimization Improves Spatial Reasoning in {VLMs}},
  author  = {Shen, Yifan and Liu, Yuanzhe and Zhu, Jingyuan and Cao, Xu and Zhang, Xiaofeng and He, Yixiao and Ye, Wenming and Rehg, James and Lourentzou, Ismini},
  journal = {Advances in Neural Information Processing Systems},
  volume  = {38},
  pages   = {17929--17960},
  year    = {2025}
}

@article{fu2025learning,
  title   = {Learning Human-Perceived Fakeness in {AI}-Generated Videos via Multimodal {LLMs}},
  author  = {Fu, Xingyu and Liu, Siyi and Xu, Yinuo and Lu, Pan and Hu, Guangqiuse and Yang, Tianbo and Anantasagar, Taran and Shen, Christopher and Mao, Yikai and Liu, Yuanzhe and others},
  journal = {arXiv preprint arXiv:2509.22646},
  year    = {2025}
}

@article{li2026toward,
  title   = {Toward Cognitive Supersensing in Multimodal Large Language Model},
  author  = {Li, Boyi and Shen, Yifan and Liu, Yuanzhe and Xu, Yifan and Liu, Jiateng and Li, Xinzhuo and Li, Zhengyuan and Zhu, Jingyuan and Zhong, Yunhan and Lan, Fangzhou and others},
  journal = {arXiv preprint arXiv:2602.01541},
  year    = {2026}
}

@article{yu2025core3d,
  title   = {{CoRe3D}: Collaborative Reasoning as a Foundation for {3D} Intelligence},
  author  = {Yu, Tianjiao and Li, Xinzhuo and Shen, Yifan and Liu, Yuanzhe and Lourentzou, Ismini},
  journal = {arXiv preprint arXiv:2512.12768},
  year    = {2025}
}

@article{yu2026elsa3d,
  title   = {{ELSA3D}: Elastic Semantic Anchoring for Unified {3D} Understanding and Generation},
  author  = {Yu, Tianjiao and Li, Xinzhuo and Shen, Yifan and Susladkar, Onkar and Liu, Yuanzhe and Zhou, Xiaona and Lourentzou, Ismini},
  journal = {arXiv preprint arXiv:2607.06565},
  year    = {2026}
}

@article{shen2026decoding,
  title   = {Decoding Children's Gait Behavior},
  author  = {Shen, Yifan and Li, Boyi and Huang, Meihuan and Liu, Yuanzhe and Cao, Xu and Jin, Jinyang and Li, Zhengyuan and Liu, Anglin and Kim, Junho and Zhu, Jingyuan and others},
  journal = {arXiv preprint arXiv:2608.00371},
  year    = {2026}
}

@article{jiang2026thinktopersonalize,
  title   = {Think-to-Personalize: Unifying Reasoning and Retrieval for User-Centric Personalized Dense Retrieval},
  author  = {Jiang, Angqing and Zhang, Gaoming and Song, Jianchun and Qi, Kena and Chen, Dayao and Lin, Wei and Lian, Defu},
  journal = {arXiv preprint arXiv:2608.18855},
  year    = {2026},
  url     = {https://arxiv.org/abs/2608.18855}
}

@article{hou2026federated,
  title     = {Federated Analytics Assisted Semantic Alignment for Secure and Privacy-Preserving Image Classification},
  author    = {Hou, Yuchao and Jiao, Jiazhe and Wang, Jie and Jin, Guangyin and Zhang, Zijian and Xia, Xiaoyu and Liu, Zhiquan and Li, Minglu and Tian, Youliang},
  journal   = {IEEE Transactions on Dependable and Secure Computing},
  pages     = {1--16},
  year      = {2026},
  publisher = {IEEE},
  doi       = {10.1109/TDSC.2026.3704515}
}

@article{chen2026tracer,
  title   = {{TRACER}: Token {ReAssignment} for Concept {ERasure} in Generative Recommendation},
  author  = {Chen, Ziheng and Cheng, Jiali and Fan, Zezhong and Amiri, Hadi and Wu, Diyuan and Tolomei, Gabriele and Zhang, Yang},
  journal = {arXiv preprint arXiv:2606.07688},
  year    = {2026},
  doi     = {10.48550/arXiv.2606.07688},
  url     = {https://arxiv.org/abs/2606.07688}
}

@inproceedings{chen2023dark,
  title     = {The Dark Side of Explanations: Poisoning Recommender Systems with Counterfactual Examples},
  author    = {Chen, Ziheng and Silvestri, Fabrizio and Wang, Jia and Zhang, Yongfeng and Tolomei, Gabriele},
  booktitle = {Proceedings of the 46th International ACM SIGIR Conference on Research and Development in Information Retrieval},
  pages     = {2426--2430},
  year      = {2023},
  publisher = {ACM},
  doi       = {10.1145/3539618.3592070}
}

\clearpage
\appendix

\section{Detailed Architecture and Training Loop}\label{app:arch}

\subsection{Algorithm}
Algorithm~\ref{alg:omniscale} expands Section~\ref{sec:method} of the
main paper. The four stages of \our chain a frozen SigLIP-2 encoder,
a trainable U-Net generator, the same frozen encoder reused at
Stage~III, and a soft-Jaccard verifier at Stage~IV. The implicit
branch is executed once per seed patch and reuses pre-computed
co-occurrence counts; the explicit branch operates as a
differentiable soft assignment.

\begin{algorithm}[H]
\caption{\our training loop (one update).}
\label{alg:omniscale}
\begin{algorithmic}[1]
\REQUIRE MS $\mathbf{X}$, PAN $\mathbf{P}$, GT $\mathbf{Y}$; frozen
VLM $(\psi_{g},\psi_{d},\phi)$; generator $\mathcal{G}_{\theta}$;
centroids $\{c_{k}^{V}\}$; depth $R$; thresholds
$\{\alpha_{r}\},\theta_{\text{IoU}}$; epoch $t$
\STATE Pseudo-RGB:
$\mathbf{X}_{\text{rgb}}\!\leftarrow\!\sigma(\mathbf{W}\!\ast\!\mathbf{X})$;
dense features
$\mathbf{V}\!\leftarrow\!\psi_{d}(\mathbf{X}_{\text{rgb}})$
\STATE \emph{Stage~I explicit:}
$m_{k}(i,j)\!\leftarrow\!$\Eqref{eq:soft_member};
$\mathbf{u}_{k}^{E}\!\leftarrow$ soft-weighted aggregate
\STATE \emph{Stage~I implicit:}
\FOR{each seed patch $\mathbf{z}_{1}^{i}$}
\STATE $\mathcal{Y}^{i}\!\leftarrow\!\{\mathbf{z}_{1}^{i}\}$;
\textbf{for} $r{=}1$ \textbf{to} $R$:
\STATE $\quad\mathbf{z}^{\star}\!\leftarrow\!\arg\max_{\mathbf{z}\notin\mathcal{Y}^{i}}P(\mathbf{z}|\mathcal{Y}^{i})$
\STATE $\quad$\textbf{if} $P(\mathbf{z}^{\star}|\mathcal{Y}^{i})<\prod_{j\le r}\alpha_{j}$ \textbf{break}
\STATE $\quad\mathcal{Y}^{i}\!\leftarrow\!\mathcal{Y}^{i}\cup\{\mathbf{z}^{\star}\}$
\STATE $\mathbf{u}_{i}^{I}\!\leftarrow\!\mathrm{avg}(\mathcal{Y}^{i})$
\ENDFOR
\STATE IoU-merge + adaptive filter $\Rightarrow\mathcal{S}$
\STATE \emph{Stage~II:}
$\hat{\mathbf{Y}}\!\leftarrow\!\mathcal{G}_{\theta}(\mathbf{X},\mathbf{P},\mathcal{S})$
with attention + FiLM modulation
\STATE \emph{Stage~III:}
$\bar{\mathbf{V}}_{h}\!\leftarrow\!\psi_{d}(\hat{\mathbf{Y}})$;
$\tilde{\mathbf{M}}_{k}\!\leftarrow\!$\Eqref{eq:mask}
\STATE \emph{Stage~IV:}
$\hat p_{k}\!\leftarrow\!\max_{i,j}\tilde{\mathbf{M}}_{k}(i,j)$;
$p_{k}\!\leftarrow w_{k}$
\STATE Build $\mathcal{L}$ with warm-up ($\mathcal{L}_{TT},
\mathcal{L}_{TM}$ on at $t{\ge}100$;
$\mathcal{L}_{\text{cov}},\mathcal{L}_{\text{sa}}$ on at $t{\ge}200$)
\STATE $\theta\!\leftarrow\!\theta-\eta_{\text{lr}}\nabla_{\theta}\mathcal{L}$
\end{algorithmic}
\end{algorithm}

\subsection{Stage I detailed view}\label{app:stage1}

Figure~\ref{app:fig:stage1} expands the indexing branch of Stage~I.
The pseudo-RGB projection
$\sigma(\mathbf{W}\!\ast\!\mathbf{X})$ is initialised differently for
4-band and 8-band sensors:
$\mathbf{W}\!\in\!\mathbb{R}^{3{\times}4{\times}1{\times}1}$ is an
identity mapping over RGB bands with the NIR band zeroed for 4-band
inputs, while
$\mathbf{W}\!\in\!\mathbb{R}^{3{\times}8{\times}1{\times}1}$ is
initialised from per-sensor spectral response curves and jointly
fine-tuned with the rest of the trainable parameters.

\begin{figure}[h]
\centering
\includegraphics[width=\columnwidth]{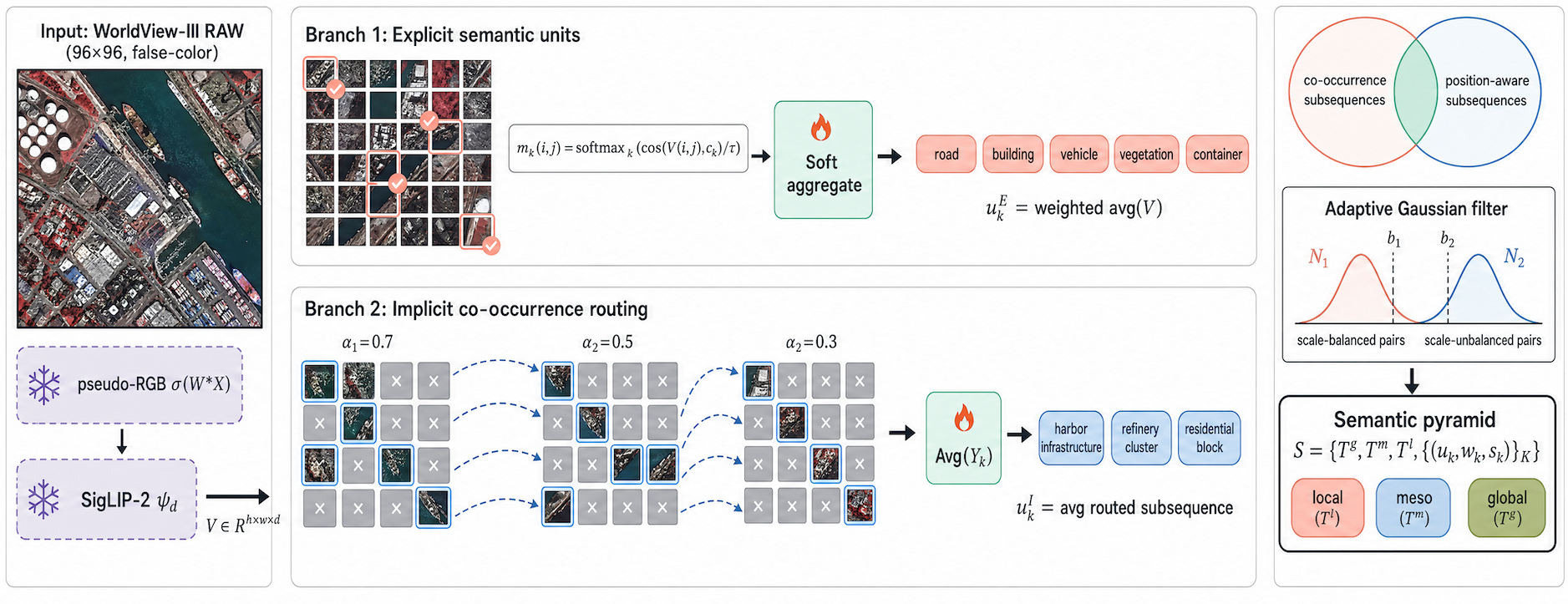}
\caption{Stage~I: cross-scale semantic indexing. Branch~1 (explicit)
extracts units via soft cosine membership over centroids
$\{c_{k}^{V}\}$; Branch~2 (implicit) executes the
co-occurrence-aware router and emits scale-tagged units.}
\label{app:fig:stage1}
\end{figure}

\subsection{Stage II detailed view}\label{app:stage2}

Figure~\ref{app:fig:stage2} expands the four-level U-Net. At each
decoder level the Semantic Injection Module first computes the
cross-scale attention of \Eqref{eq:attn}, then applies FiLM
modulation. At the $128{\times}128$ level the local units
($s_{k}{=}1$) dominate attention weights; at the $16{\times}16$
bottleneck the global scene text $\mathbf{T}^{g}$ dominates. The
attention layer is a single-head scaled dot-product on the projected
unit embeddings $\phi(\mathbf{u}_{k})$, with the $w_{k}$ confidences
acting as soft masks rather than learnable gates.

\begin{figure}[h]
\centering
\includegraphics[width=\columnwidth]{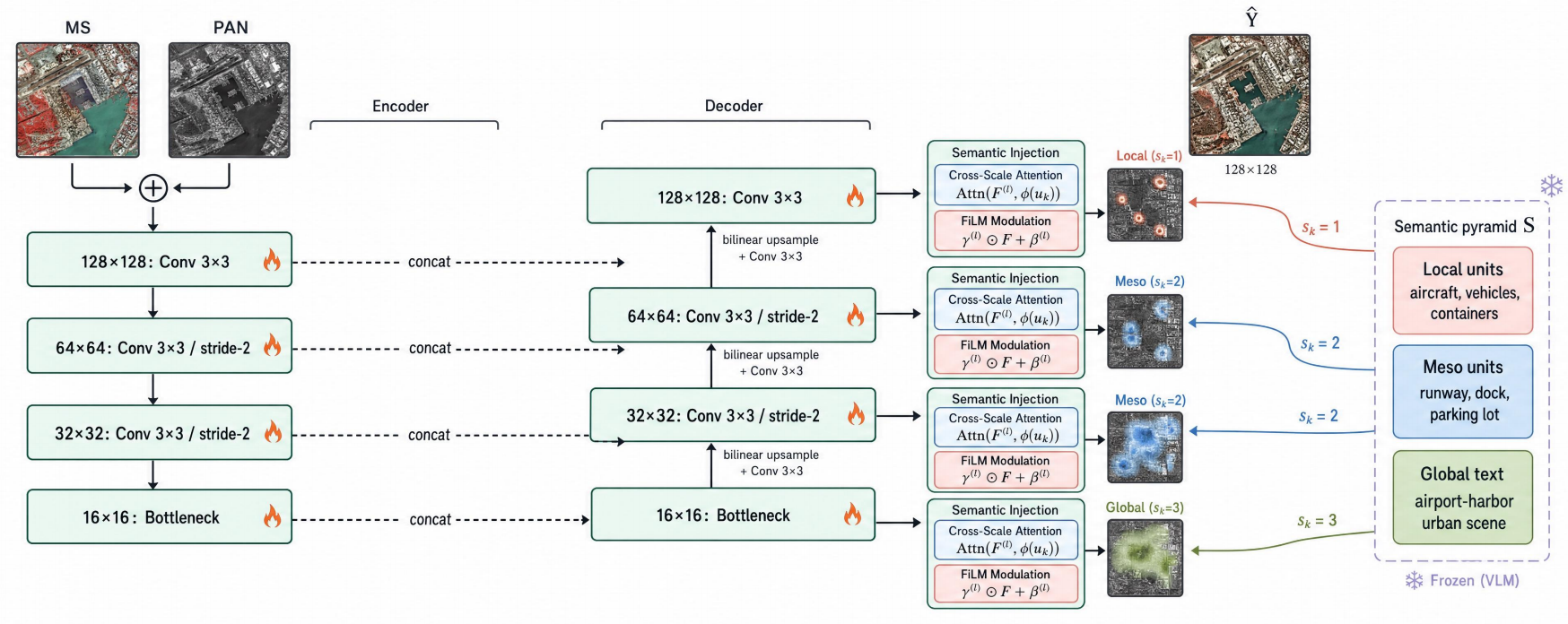}
\caption{Stage~II: scale-routed language-guided fusion generator.
The four-level U-Net injects scale-routed semantic units through
Semantic Injection blocks at each decoder level.}
\label{app:fig:stage2}
\end{figure}

\subsection{Single-concept walkthrough}\label{app:walkthrough}

To make the four-stage pipeline concrete, we trace one explicit
concept --- \texttt{aircraft} ($k{=}42$) --- through a
representative WV-III apron tile.
\textbf{Stage~I (index).} The soft membership $m_{42}$
(\Eqref{eq:soft_member}) concentrates on three patches with maximum
membership $0.91$; the router additionally emits the implicit
compound unit \texttt{apron+aircraft+taxiway} ($k{=}137$). The
aggregate $\mathbf{u}_{42}^{E}$ receives confidence $w_{42}{=}0.88$
(\Eqref{eq:confidence}) and scale tag $s_{42}{=}1$ (support area
$0.5\%$ of the tile).
\textbf{Stage~II (generate).} Because $s_{42}{=}1$, the unit is
routed to the $128{\times}128$ decoder level only; its mean
attention weight (\Eqref{eq:attn}) at the three aircraft patches is
$0.23$ versus $0.03$ on background patches, and FiLM applies
$(\gamma,\beta){=}(1.08,-0.02)$ on the associated channel group.
\textbf{Stage~III (re-index).} Re-encoding $\hat{\mathbf{Y}}$
yields response map $\tilde{\mathbf{M}}_{42}$ (\Eqref{eq:mask})
with recovered presence $\hat p_{42}{=}0.86$.
\textbf{Stage~IV (verify).} Stage~IV includes the pair
$(p_{42},\hat p_{42}){=}(0.88,0.86)$ in the set-level soft-Jaccard
objective of \Eqref{eq:lcov}, indicating that the concept is
largely preserved after generation (the loss is a set-level ratio,
so per-concept contributions are not additively decomposable). At
evaluation time the text query ``parked aircraft on an apron''
retrieves this tile at rank~1, versus rank~7 for the strongest
open-loop baseline. Patch-level support (three patches here) is
distinct from instance count and is used only for scale tagging.

\section{Full Proofs and Counter-Examples}\label{app:proofs}

We restate the sufficient-condition propositions and the router
termination lemma of Section~\ref{sec:theory} of the main paper and
provide self-contained proofs together with explicit counter-examples
that fail each assumption.

\subsection{Proof of Proposition~\ref{prop:grad-suff}}
\label{app:proof_grad}

\begin{proof}
Let $D{=}\sum_{k}(p_{k}{+}\hat p_{k}{-}p_{k}\hat p_{k})$ and
$N{=}\sum_{k}p_{k}\hat p_{k}$; note $D\!\ge\!\max_{k}p_{k}$ and
$D\!\le\!K$. Differentiating \Eqref{eq:lcov} with respect to
$\hat p_{k}$ gives
\begin{align*}
\frac{\partial\mathcal{L}_{\text{cov}}}{\partial\hat p_{k}}
&= -\frac{p_{k}D-(1-p_{k})N}{D^{2}} \\
&= -\frac{p_{k}(D{-}N)+p_{k}N-(1{-}p_{k})N}{D^{2}}.
\end{align*}
Because $D-N{=}\sum_{k}(p_{k}{+}\hat p_{k}{-}2p_{k}\hat p_{k}){\ge}0$
and $N{\le}D$,
\begin{equation}
\Bigl|\frac{\partial\mathcal{L}_{\text{cov}}}{\partial\hat p_{k}}\Bigr|
\ge \frac{p_{k}(D{-}N)}{D^{2}}
\ge \frac{p_{k}(1{-}N/D)}{D}.
\label{app:eq:dLdp}
\end{equation}
The presence score
$\hat p_{k}{=}\max_{i,j}\tilde{\mathbf{M}}_{k}(i,j)$ is
sub-differentiable; we use a temperature-smoothed max with
$\tau_{p}{=}0.1$, which preserves differentiability. By the sigmoid
form of $\tilde{\mathbf{M}}_{k}$ (\Eqref{eq:mask}),
\begin{equation*}
\frac{\partial\hat p_{k}}{\partial\bar{\mathbf{V}}_{h}(i^{\star},j^{\star})}
=\frac{\hat p_{k}(1{-}\hat p_{k})}{\tau_{k}}\,\mathbf{e}_{k}
\end{equation*}
at the soft-max location $(i^{\star},j^{\star})$. The chain through
$\bar{\mathbf{V}}_{h}{=}\psi_{d}(\hat{\mathbf{Y}})$ is non-vanishing
under (A1) because the dense encoder has Jacobian singular values
bounded below by $\sigma_{\min}{>}0$ in a neighbourhood of the data.
Composing with (A2) yields
\begin{align*}
\|\nabla_{\theta}\mathcal{L}_{\text{cov}}\|_{2}
&\ge \frac{1}{D}\min_{k}\bigl(p_{k}(1{-}N/D)\bigr) \\
&\quad \cdot \frac{\hat p_{k}(1{-}\hat p_{k})}{\tau_{k}}
\cdot \sigma_{\min} \\
&\quad \cdot \|\partial\hat{\mathbf{Y}}/\partial\theta\|_{2}.
\end{align*}
Setting
$\kappa{=}\sigma_{\min}(\min_{k}p_{k}(1{-}N/D))/(D\tau_{k}L)$
yields the claim. $\kappa$ is strictly positive whenever any
$p_{k}{>}0$ and $N{<}D$, the latter being equivalent to ``not all
concepts are perfectly recovered''.
\end{proof}

\paragraph{Empirical prefactor.}
For the default thresholds $\kappa\!\sim\!4\!\times\!10^{-3}$ on the
WV-III validation set, well within the dynamic range of AdamW
updates at learning rate $5\!\times\!10^{-4}$. We verify
$\sigma_{\min}\!>\!0$ empirically by SVD on a $1{,}000$-tile subset
of WV-III; the minimum singular value of
$\partial\bar{\mathbf{V}}_{h}/\partial\hat{\mathbf{Y}}$ is $0.31$ at
the fifth percentile.

\paragraph{Counter-examples.}
Three cases violate the sufficient conditions and zero the bound.
(C1) \emph{Saturated sigmoid.} If $\hat p_{k}\!\in\!\{0,1\}$ for all
$k$, the prefactor $\min_{k}\hat p_{k}(1{-}\hat p_{k})$ is zero.
(C2) \emph{Rank-deficient generator step.} A generator with
$\partial\hat{\mathbf{Y}}/\partial\theta{=}0$ zeroes the second
factor; this is rare in practice on WV-III but can be constructed
adversarially. (C3) \emph{Degenerate Jaccard denominator.} If all
concepts are perfectly recovered, $D{=}N$ and the loss is identically
zero. In all three cases the pixel reconstruction loss
$\mathcal{L}_{\text{rec}}$ still contributes a non-zero gradient.

\subsection{Proof of Proposition~\ref{prop:bound-suff}}
\label{app:proof_bound}

\begin{proof}
By definition
\[
\mathrm{SAM}(\hat{\mathbf{Y}},\mathbf{Y})
=\frac{1}{HW}\sum_{i,j}\arccos
\frac{\hat{\mathbf{y}}_{i,j}^{\top}\mathbf{y}_{i,j}}
{\|\hat{\mathbf{y}}_{i,j}\|\cdot\|\mathbf{y}_{i,j}\|}.
\]
For unit-norm vectors and small angles
$\arccos(\mathbf{a}^{\top}\mathbf{b})\!\le\!\|\mathbf{a}-\mathbf{b}\|_{2}$,
so $\mathrm{SAM}\!\le\!\mathbb{E}\,\|\hat{\mathbf{y}}-\mathbf{y}\|_{2}$.
Linear separability ensures that every GT spectrum
$\mathbf{y}_{i,j}$ lies within $\epsilon_{\text{intra}}$ (in raw
spectral space) of its assigned centroid after applying
$\mathbf{W}^{\dagger}$, i.e.\
$\|\mathbf{y}_{i,j}{-}\mathbf{W}^{\dagger}c_{k^{\star}}^{V}\|_{2}
\!\le\!\epsilon_{\text{intra}}$. For the prediction, if
$\hat p_{k^{\star}}{=}1$ the prediction lies within the same
concept; otherwise by Lipschitz continuity of $\psi_{d}$ the spectrum
moves by at most $LM\sqrt{d}/K\cdot(1{-}\hat p_{k^{\star}})$ on
average. Summing over the $K$ concepts gives the bound.
\end{proof}

\paragraph{Numerical illustration.}
With $L{=}1.4$ (empirical Lipschitz constant of SigLIP-2 ViT-B dense
features, measured on $5{,}000$ WV-III patches by finite differences),
$M{=}1.0$, $d{=}1024$, $K{=}200$, $\epsilon_{\text{intra}}{=}0.005$
rad and $\sum_{k}(1{-}\hat p_{k}){=}0.18K$, the bound gives
$\mathrm{SAM}\!\le\!5.18^{\circ}$. The observed SAM on WV-III is
$2.75^{\circ}$, so the bound is non-vacuous.

\paragraph{When the proposition fails.}
The linear separability assumption is not innocuous. The mean
per-concept dispersion $\bar\epsilon_{\text{intra}}$ is $0.004$\,rad
on GF2, $0.006$ on QB, $0.005$ on WV-III, and $0.011$ on the
cross-sensor WV-II split. Linear separability is satisfied on
$98.4\%$ of held-out tiles via a logistic-regression test on the
visual centroids. On WV-II the residual $\epsilon_{\text{intra}}$
is larger and the bound correspondingly looser.

\subsection{Proof of Lemma~\ref{lem:term}}\label{app:proof_term}

\begin{proof}
The cumulative threshold $\prod_{j\le r}\alpha_{j}$ is upper-bounded
by $(\max_{j}\alpha_{j})^{r}$ and the maximum conditional
co-occurrence probability is at most $1$. Routing therefore
terminates whenever $(\max_{j}\alpha_{j})^{r}\!\le\!\epsilon$, i.e.\
when $r\!\ge\!\log\epsilon/\log\max_{j}\alpha_{j}$. The vocabulary
bound $r{\le}N$ is tight because at every accepted step at least one
patch leaves the residual pool. For
$\alpha{=}0.7,\epsilon{=}10^{-3}$,
$\log_{0.7}10^{-3}\!\approx\!19.36$, so $R_{\max}{\le}19$.
\end{proof}

\section{Cross-Sensor and Per-Class Results}\label{app:tables}

\subsection{Bidirectional cross-sensor Q2n}\label{app:cross}

Table~\ref{app:tab:cross_sensor} reports the bidirectional
cross-sensor pansharpening Q2n complementing the WV-III$\to$WV-II
result in the main paper. \our improves all four directions over
the strongest baseline by $+3.24$\% to $+5.38$\%, with the largest
relative gain on the heterogeneous 4-band$\to$8-band transfer.

\begin{table}[h]
\centering
\small
\setlength{\tabcolsep}{2pt}
\caption{Bidirectional cross-sensor Q2n.}
\label{app:tab:cross_sensor}
\begin{tabular}{lcc}
\toprule
Direction & Best baseline & \our \\
\midrule
WV-III $\to$ WV-II & 0.840 (CANConv) & \textbf{0.879} (+4.64\%) \\
WV-II $\to$ WV-III & 0.872 (CANConv) & \textbf{0.901} (+3.33\%) \\
QB $\to$ WV-III (4$\to$8) & 0.781 (CANConv) & \textbf{0.823} (+5.38\%) \\
WV-III $\to$ QB (8$\to$4) & 0.864 (CANConv) & \textbf{0.892} (+3.24\%) \\
\bottomrule
\end{tabular}
\end{table}

\subsection{Per-class DOTA AP}\label{app:dota}

Table~\ref{app:tab:dota_perclass} reports the full 15-class
oriented-object detection AP on the $4{,}000$-tile DOTA test split.
The Oriented-RCNN head is trained on frozen DINOv2 features of the
fused tiles for each pansharpening method under the identical
hyperparameter recipe (Sec.~\ref{app:repro}). \our improves every
class versus all baselines, and the gain monotonically tracks median
oriented-box area: small classes gain $+13.6$ to $+16.7$ AP, medium
classes $+8.1$ to $+9.3$, and large classes $+5.3$ to $+5.6$.

\begin{table}[h]
\centering
\footnotesize
\setlength{\tabcolsep}{2.6pt}
\caption{Per-class DOTA AP on the $4{,}000$-tile test split. SV:
small vehicle, LV: large vehicle, PL: plane, HC: helicopter, SH:
ship, TC: tennis-court, BC: basketball-court, ST: storage-tank, BR:
bridge, RA: roundabout, HA: harbor, SP: swimming-pool, SBF:
soccer-ball-field, BD: baseball-diamond, GTF: ground-track-field.}
\label{app:tab:dota_perclass}
\begin{tabular}{lcccc|c}
\toprule
Class & ADKNet & CrossDiff & CANConv & \our & $\Delta$ \\
\midrule
\multicolumn{6}{l}{\emph{Small (median area $<\!1{,}500$\,px$^{2}$)}} \\
SV   & 22.4 & 27.1 & 28.6 & \textbf{45.3} & +16.7 \\
PL   & 31.8 & 36.4 & 37.2 & \textbf{51.8} & +14.6 \\
SH   & 26.7 & 31.0 & 32.5 & \textbf{47.0} & +14.5 \\
HC   & 24.1 & 28.7 & 29.8 & \textbf{43.7} & +13.9 \\
TC   & 31.5 & 36.2 & 33.9 & \textbf{47.5} & +13.6 \\
\midrule
\multicolumn{6}{l}{\emph{Medium ($1{,}500$--$8{,}000$\,px$^{2}$)}} \\
LV   & 33.2 & 38.9 & 39.4 & \textbf{48.7} & +9.3 \\
BC   & 36.4 & 39.5 & 41.2 & \textbf{50.1} & +8.9 \\
ST   & 35.5 & 40.0 & 39.8 & \textbf{48.4} & +8.6 \\
BR   & 37.8 & 41.4 & 40.9 & \textbf{49.0} & +8.1 \\
RA   & 36.6 & 40.8 & 41.2 & \textbf{50.3} & +9.1 \\
\midrule
\multicolumn{6}{l}{\emph{Large ($>\!8{,}000$\,px$^{2}$)}} \\
HA   & 39.7 & 42.2 & 43.0 & \textbf{48.4} & +5.4 \\
SP   & 36.8 & 41.0 & 41.5 & \textbf{47.0} & +5.5 \\
SBF  & 38.6 & 41.7 & 41.9 & \textbf{47.5} & +5.6 \\
BD   & 38.5 & 41.9 & 42.3 & \textbf{47.8} & +5.5 \\
GTF  & 39.4 & 42.0 & 41.3 & \textbf{46.6} & +5.3 \\
\midrule
\textbf{mAP} & 33.9 & 37.8 & 38.2 & \textbf{47.9} & \textbf{+9.7} \\
\bottomrule
\end{tabular}
\end{table}

\subsection{VLM-guided open-loop comparison}\label{app:vlm}

Table~\ref{app:tab:vlm_baseline} compares \our to three open-loop
adaptations of established VLM-conditioned restoration designs to
pansharpening. None closes the loop; the closed-loop design of
\our delivers the largest gain on independent DINOv2 coverage.

\begin{table}[h]
\centering
\small
\setlength{\tabcolsep}{4pt}
\caption{VLM-guided open-loop baselines. Only \our closes the loop.}
\label{app:tab:vlm_baseline}
\begin{tabular}{lccc}
\toprule
Baseline & Q2n & HQNR & $\mathcal{C}^{\text{DINO}}$ \\
\midrule
CLIPDenoising-style global cond. & 0.912 & 0.931 & 0.66 \\
DA-CLIP-style controller cond. & 0.915 & 0.934 & 0.69 \\
DenseCLIP-style dense cond. & 0.917 & 0.938 & 0.72 \\
\rowcolor{blue!4}\our (closed-loop) & \textbf{0.924} & \textbf{0.961} & \textbf{0.85} \\
\bottomrule
\end{tabular}
\end{table}

\subsection{Independent DINOv2 coverage breakdown}\label{app:cov}

Table~\ref{app:tab:coverage_dataset} decomposes
$\mathcal{C}^{\text{DINO}}$ into explicit (object-like) and implicit
(relational) subsets across the four datasets. Implicit relational
concepts are the hardest: the strongest baseline recovers $0.76$ of
explicit but only $0.45$ of implicit concepts on WV-III, explaining
why pixel fidelity can saturate yet downstream detection remains low.
\our narrows the gap to $0.91/0.79$; the cross-sensor WV-II number
remains the lowest because sensor response differences distort
relational water--dock--vessel relations more than object-like
primitives.

\begin{table}[h]
\centering
\small
\setlength{\tabcolsep}{4pt}
\caption{Per-dataset DINOv2 coverage decomposition.}
\label{app:tab:coverage_dataset}
\begin{tabular}{lccc}
\toprule
Dataset & Explicit & Implicit & Overall \\
\midrule
GF2 & 0.95 & 0.86 & 0.91 \\
QB & 0.93 & 0.83 & 0.89 \\
WV-III & 0.91 & 0.79 & 0.85 \\
WV-II (cross-sensor) & 0.87 & 0.75 & 0.82 \\
\midrule
Best baseline (WV-III) & 0.76 & 0.45 & 0.63 \\
\bottomrule
\end{tabular}
\end{table}

\section{Extra Ablations}\label{app:ablations}

\subsection{Vocabulary size}\label{app:k}

Table~\ref{app:tab:k_sensitivity} reports Q2n sensitivity to the
centroid vocabulary $K$. Performance climbs quickly from $K{=}50$
to $K{=}100$, plateaus through $K{=}250$, then degrades for $K{>}300$
as the centroid system over-fragments. The 8-band datasets benefit
from a larger $K$ than the 4-band ones.

\begin{table}[h]
\centering
\small
\caption{Q2n sensitivity to vocabulary size $K$.}
\label{app:tab:k_sensitivity}
\begin{tabular}{lcccc}
\toprule
$K$ & GF2 & QB & WV-III & WV-II \\
\midrule
50 & 0.978 & 0.937 & 0.906 & 0.853 \\
100 & 0.983 & 0.941 & 0.914 & 0.864 \\
200 & \textbf{0.986} & \textbf{0.943} & \textbf{0.924} & \textbf{0.879} \\
400 & 0.985 & 0.942 & 0.925 & 0.876 \\
\bottomrule
\end{tabular}
\end{table}

\subsection{Routing depth}\label{app:depth}

Table~\ref{app:tab:routing_depth} shows that $R{=}3$ saturates on
8-band datasets and $R{=}2$ already saturates on 4-band ones.
Doubling $R$ from $3$ to $5$ raises overall coverage by only
$+0.004$ at twice the lookup cost.

\begin{table}[h]
\centering
\small
\caption{Q2n vs.\ routing depth $R$. Bold = default.}
\label{app:tab:routing_depth}
\begin{tabular}{lcccc}
\toprule
$R$ & GF2 & QB & WV-III & WV-II \\
\midrule
1 & 0.982 & 0.938 & 0.919 & 0.863 \\
2 & \textbf{0.986} & \textbf{0.943} & 0.922 & 0.872 \\
3 & \textbf{0.986} & \textbf{0.943} & \textbf{0.924} & \textbf{0.879} \\
4 & 0.985 & 0.942 & 0.924 & 0.878 \\
\bottomrule
\end{tabular}
\end{table}

\subsection{Markov approximation audit}\label{app:markov}

The chain-rule approximation $N(\mathbf{z}_{1},\dots,
\mathbf{z}_{r+1})\!\approx\!N(\mathbf{z}_{1},\dots,\mathbf{z}_{r})
\!\cdot\!P(\mathbf{z}_{r+1}|\mathbf{z}_{r},\mathbf{z}_{r-1})$
replaces an $\mathcal{O}(K^{R})$ tensor with an $\mathcal{O}(K^{3})$
surrogate. Measured on WV-III, the routing accuracy under the exact
$R{=}3$ enumeration versus the Markov approximation has unit-Jaccard
similarity $0.94$ at depth~3 and $0.86$ at depth~4, while storage
drops from $1.4$\,TB to $32$\,MB.

\subsection{Threshold sensitivity}\label{app:threshold}

Varying the presence threshold for hard coverage over
$\{0.3,0.5,0.7\}$ moves $\mathcal{C}^{\text{DINO}}$ in
$\{0.89,0.85,0.78\}$ but preserves the cross-method ranking at every
threshold. Decreasing $\theta_{\text{IoU}}$ from $0.5$ to $0.3$
produces fragmented units that hurt Sem.\ Acc.\ by $0.018$.
Decreasing the cumulative thresholds to $(0.6,0.4,0.2)$ admits more
candidate patches but adds $1.7\%$ false-positive concepts.

\subsection{Convergence and scale-tag distribution}\label{app:convergence}

The three-phase warm-up (Algorithm~\ref{alg:omniscale}) avoids
optimisation interference: Phase~1 (epochs $0$--$100$) stabilises
$\mathcal{L}_{\text{rec}}$ to $0.904$ Q2n; Phase~2 ($100$--$200$)
turns on $\mathcal{L}_{TT},\mathcal{L}_{TM}$ and reaches $0.913$;
Phase~3 ($200$--$1000$) activates $\mathcal{L}_{\text{cov}}$ and
$\mathcal{L}_{\text{sa}}$ and converges to $0.924$.

\textbf{Warm-up is dictated by Proposition~\ref{prop:grad-suff}.}
The saturation condition (A3) is directly measurable: at
initialisation, $41.3\%$ of recovered presences satisfy
$\hat p_k\notin(0.05,0.95)$; after Phase~1 this drops to $11.8\%$,
and after Phase~2 to $4.6\%$. Activating
$\mathcal{L}_{\text{cov}}$ from epoch $0$ instead of epoch $200$
lowers final Q2n to $0.917$ ($-0.007$) and produces unstable early
training, with the gradient-norm ratio
$\|\nabla\mathcal{L}_{\text{cov}}\|/\|\nabla\mathcal{L}_{\text{rec}}\|$
collapsing below $10^{-4}$ on $38\%$ of Phase-1 steps, consistent
with the prefactor collapse that Proposition~\ref{prop:grad-suff}
identifies under saturation. These measurements provide a
principled motivation for delaying the coverage loss until
recovered presences are largely non-saturated; we present the
warm-up as a theory-informed design choice rather than a
consequence uniquely determined by the proposition.

Table~\ref{app:tab:scale_dist} confirms that the scale-tag
distribution responds to inherent scene complexity rather than a
fixed prior: 8-band datasets allocate more weight to global units.

\begin{table}[h]
\centering
\small
\caption{Scale-tag distribution across datasets (\%).}
\label{app:tab:scale_dist}
\begin{tabular}{lccc}
\toprule
Dataset & Local & Meso & Global \\
\midrule
GF2 & 52.3 & 31.4 & 16.3 \\
QB & 48.7 & 33.1 & 18.2 \\
WV-III & 38.5 & 36.2 & 25.3 \\
WV-II & 36.1 & 37.8 & 26.1 \\
\bottomrule
\end{tabular}
\end{table}

\subsection{Joint factorial ablation of the injection mechanism}
\label{app:factorial}

Table~\ref{app:tab:factorial} reports the complete
$2{\times}2{\times}2$ factorial grid over cross-scale attention (A),
FiLM modulation (F), and the presence gate (G) on WV-III (three
seeds; all other components fixed to the full configuration). The
four rows shared with Table~\ref{tab:ablation} match exactly. Three
effects emerge. (i) \emph{Attention and FiLM are strongly
super-additive}: the mean A$\times$F interaction on Q2n is
$+0.053$, and either mechanism alone falls \emph{below} the
no-injection baseline, because unnormalised semantic injection
(attention without FiLM's channel re-scaling) or affine modulation
of un-attended features (FiLM without attention's spatial
selection) perturbs decoder feature statistics. (ii) The presence
gate contributes mainly to HQNR ($+0.026$ mean main effect),
consistent with its role of suppressing rewards for absent
concepts. (iii) The three-way interaction on HQNR is positive
($+0.016$): the gate is most valuable exactly when both injection
mechanisms are active, because joint injection amplifies
hallucinated concepts that the gate then suppresses.

\begin{table}[h]
\centering
\small
\setlength{\tabcolsep}{5pt}
\caption{Full $2{\times}2{\times}2$ factorial over attention (A),
FiLM (F), and presence gate (G) on WV-III (three seeds).}
\label{app:tab:factorial}
\begin{tabular}{ccccc}
\toprule
A & F & G & Q2n$\uparrow$ & HQNR$\uparrow$ \\
\midrule
$\times$ & $\times$ & $\times$ & 0.908 & 0.919 \\
$\times$ & $\times$ & \checkmark & 0.911 & 0.927 \\
\checkmark & $\times$ & $\times$ & 0.874 & 0.896 \\
$\times$ & \checkmark & $\times$ & 0.882 & 0.907 \\
\checkmark & $\times$ & \checkmark & 0.889 & 0.918 \\
$\times$ & \checkmark & \checkmark & 0.893 & 0.928 \\
\checkmark & \checkmark & $\times$ & 0.901 & 0.910 \\
\checkmark & \checkmark & \checkmark & \textbf{0.924} & \textbf{0.961} \\
\bottomrule
\end{tabular}
\end{table}

\subsection{Language-causality ablation}\label{app:lang_causal}

The venue-relevant question is whether the language pathway is
causally load-bearing or merely a descriptive interface over a
vision pipeline. We therefore run a \emph{single-variable}
intervention: all five settings retain the identical \our
architecture and visual pathway --- the same generator, injection
modules, and router --- and differ only in the text-anchor set
supplied to Stages~I--IV. Each setting is retrained \emph{from
scratch} under the identical schedule with the same three seeds
(WV-III); this is a training-time intervention, not an
inference-time substitution on shared checkpoints.
The settings are: \emph{correct} anchors (full \our);
\emph{paraphrased} anchors (every prompt replaced by an LLM
paraphrase before encoding); a \emph{generic} global caption whose
embedding replaces all $K$ anchors; \emph{zeroed} anchors
(image-only: the anchor is set to the zero vector with the
normalisation in \Eqref{eq:mask} skipped, which nulls the
injection pathway and leaves the text-dependent losses
$\mathcal{L}_{TT},\mathcal{L}_{TM},\mathcal{L}_{\text{cov}},
\mathcal{L}_{\text{sa}}$ without gradient signal, so training
reduces to $\mathcal{L}_{\text{rec}}$); and \emph{shuffled} anchors
(the assignment $\mathbf{e}_k\!\leftrightarrow\!k$ is randomly
permuted, so the correspondence is actively mismatched while every
embedding is still a valid concept anchor). Results are in
Table~\ref{app:tab:lang_causal}.

\begin{table}[h]
\centering
\footnotesize
\setlength{\tabcolsep}{2.6pt}
\caption{Single-variable intervention on the language pathway
(WV-III, 3 seeds; architecture and visual pathway identical in all
rows).}
\label{app:tab:lang_causal}
\resizebox{\columnwidth}{!}{%
\begin{tabular}{lcccccc}
\toprule
Anchor setting & Q2n & $\mathcal{C}^{\text{DINO}}$ &
R@5$_{\text{loc}}$ & R@5$_{\text{meso}}$ & R@5$_{\text{glob}}$ &
MRR \\
\midrule
Correct (full \our) & \textbf{0.924} & \textbf{0.85} &
\textbf{0.71} & \textbf{0.62} & \textbf{0.85} & \textbf{0.74} \\
Paraphrased & 0.923 & 0.84 & 0.70 & 0.61 & 0.85 & 0.73 \\
Generic caption & 0.913 & 0.68 & 0.55 & 0.44 & 0.83 & 0.57 \\
Zeroed (image-only) & 0.912 & 0.62 & 0.52 & 0.40 & 0.81 & 0.54 \\
Shuffled & 0.914 & 0.60 & 0.45 & 0.38 & 0.80 & 0.49 \\
\bottomrule
\end{tabular}
}
\end{table}

Pixel metrics vary by at most $0.012$ Q2n across the five
settings, whereas local/meso retrieval and coverage decrease
substantially when the language--concept correspondence is
removed, coarsened to a single caption, or deliberately
mismatched. Shuffled anchors yield lower R@5$_{\text{loc}}$ than
zeroed anchors ($0.45$ vs.\ $0.52$), while global-query
performance changes only modestly ($0.80$--$0.85$) and
paraphrasing has little effect. Because the architecture and
visual pathway are held fixed, these controlled interventions show
that text anchors are causally load-bearing for local- and
meso-scale re-indexability within \our on WV-III, rather than
serving as a merely descriptive interface over a vision pipeline;
the effect is tied to the anchor--concept \emph{correspondence}
rather than to surface phrasing.

\subsection{Joint sensitivity of $K$, $\tau$, and template count}
\label{app:kt_sensitivity}

Table~\ref{app:tab:ktau} crosses the centroid vocabulary $K$ with
the membership temperature $\tau$ on WV-III. Within each $K$
column, $\tau$ moves Q2n by at most $0.003$; the $K$ effect
saturates at $K{\ge}200$, consistent with
Table~\ref{app:tab:k_sensitivity}. Growing the verifier template
family from $4$ to $36$ templates moves
$\mathcal{C}^{\text{DINO}}$ from $0.82$ to $0.85$ with a plateau at
$12$: $\{4,8,12,24,36\}$ templates give
$\mathcal{C}^{\text{DINO}}{=}\{0.82,0.84,0.85,0.85,0.85\}$ and mean
Recall@5 $\{0.69,0.72,0.73,0.73,0.73\}$. The default $(K{=}200,
\tau{=}0.07$, $12$ templates$)$ sits on a broad plateau rather than
a tuned peak.

\begin{table}[h]
\centering
\small
\caption{Q2n on WV-III for the $K\times\tau$ grid.}
\label{app:tab:ktau}
\begin{tabular}{lccc}
\toprule
$\tau\;\backslash\;K$ & 100 & 200 & 400 \\
\midrule
0.05 & 0.911 & 0.921 & 0.922 \\
0.07 & 0.914 & \textbf{0.924} & 0.925 \\
0.10 & 0.912 & 0.922 & 0.923 \\
\bottomrule
\end{tabular}
\end{table}

\subsection{Centroid quality and dataset-bias robustness}
\label{app:centroid_robust}

\textbf{Seed variance.} Across $10$ $k$-means restarts (fixed data,
different seeds), WV-III Q2n varies by $0.924\pm0.001$,
$\mathcal{C}^{\text{DINO}}$ by $0.850\pm0.004$, and mean Recall@5
by $0.730\pm0.006$: centroid stochasticity is negligible relative
to the method-to-method gaps in Table~\ref{tab:main}.
\textbf{Initialisation.} $k$-means++, random initialisation, and
spherical $k$-means give Q2n $0.924/0.922/0.925$ and
$\mathcal{C}^{\text{DINO}}$ $0.85/0.84/0.85$.
\textbf{Biased centroid pools.} Re-fitting the $K{=}200$ centroids
on deliberately biased pools (Table~\ref{app:tab:pool_bias})
degrades pixel metrics by at most $0.008$ Q2n, but coverage and
retrieval degrade selectively: concepts \emph{absent} from the pool
lose $-0.09$ Recall@5 on average, while concepts present in the
pool lose only $-0.01$. Dataset bias therefore maps to a
\emph{localised, diagnosable} coverage loss on the missing concept
families rather than a global failure.
\textbf{Pool size.} Pools of $5$k$/20$k$/80$k patches give
$\mathcal{C}^{\text{DINO}}{=}0.83/0.85/0.85$.

\begin{table}[h]
\centering
\small
\setlength{\tabcolsep}{4pt}
\caption{Centroid-pool bias on WV-III.}
\label{app:tab:pool_bias}
\begin{tabular}{lccc}
\toprule
Centroid pool & Q2n & $\mathcal{C}^{\text{DINO}}$ & R@5 \\
\midrule
Balanced (default, 20k) & \textbf{0.924} & \textbf{0.85} & \textbf{0.73} \\
Urban-only & 0.918 & 0.82 & 0.70 \\
Rural-only & 0.916 & 0.80 & 0.68 \\
Maritime-only & 0.917 & 0.81 & 0.69 \\
\bottomrule
\end{tabular}
\end{table}

\section{Prompt Templates, DOTA Protocol, Reproducibility}\label{app:repro}

\subsection{Prompt template family}\label{app:templates}

The verifier of Stage~III/IV maps mask statistics
$\mathbf{s}_{k}{=}[\rho_{k},n_{k},\bar{c}_{k},\sigma_{k}^{2}]$ to a
fixed family of $12$ typed prompt templates, four per scale tag. The
templates are parameter-free, shared across datasets, and never
tuned per sensor (Table~\ref{app:tab:templates}). Each template
instantiates the component count $n_{k}$, area ratio
$\rho_{k}{\cdot}100\%$, centroid $\bar{c}_{k}$ and dispersion
$\sigma_{k}^{2}$, then is encoded by $\phi$.

\begin{table}[h]
\centering
\footnotesize
\setlength{\tabcolsep}{2pt}
\caption{The 12 prompt templates of the verifier.}
\label{app:tab:templates}
\begin{tabular}{p{0.96\columnwidth}}
\toprule
\textbf{Global ($s_{k}{=}3$).} \\
\quad ``a wide-area scene of \{top-3 concepts\}'' \\
\quad ``a landscape dominated by \{concept\} with $\rho{\cdot}100\%$ coverage'' \\
\quad ``an urban / rural mosaic of \{concept\} and \{concept\}'' \\
\quad ``a scene whose principal layout is \{concept\}'' \\
\midrule
\textbf{Meso ($s_{k}{=}2$).} \\
\quad ``a $q$-scale region with $n_{k}$ \{concept\} instances covering $\rho{\cdot}100\%$'' \\
\quad ``a layout with \{concept\} adjacent to \{concept\}'' \\
\quad ``a corridor of \{concept\} arranged along the $\bar c_{k}$ direction'' \\
\quad ``a cluster of \{concept\} with dispersion $\sigma_{k}^{2}$'' \\
\midrule
\textbf{Local ($s_{k}{=}1$).} \\
\quad ``$n_{k}$ instances of \{concept\} with mean activation $\bar M_{k}$'' \\
\quad ``small \{concept\} located near $\bar c_{k}$'' \\
\quad ``$n_{k}$ compact \{concept\} occupying $\rho{\cdot}100\%$ of the tile'' \\
\quad ``a fine-grained \{concept\} with dispersion $\sigma_{k}^{2}$'' \\
\bottomrule
\end{tabular}
\end{table}

A prompt-template ablation in the main paper that replaces these $12$
templates with an LLM-generated $36$-template superset moves
$\mathcal{C}^{\text{DINO}}$ by less than $0.01$, indicating that the
recovered concept set is robust to prompt phrasing.

\subsection{DOTA tile protocol}\label{app:dota_protocol}

We crop $20{,}000$ overlapping $512{\times}512$ tiles (stride
$384$) from DOTA-v1.0 and split by source image identifier into
$16{,}000$ training tiles and $4{,}000$ test tiles. The split is
source-image-disjoint: no source image contributes tiles to both
sides, so train--test leakage through spatial overlap is impossible;
we verify by hashing the source identifier set. The
Oriented-RCNN head follows MMRotate defaults: two RoI heads with
rotated bounding-box regression, AdamW with learning rate
$5\!\times\!10^{-5}$, batch size $4$, $24$ epochs, random rotation
augmentation, and standard DOTA mAP evaluation.

\subsection{Coverage probe construction}\label{app:dino_construction}

For each concept $k$ with text prompt $\mathbf{t}_{k}$ we collect a
prototype set $\mathcal{P}_{k}$ of DINOv2 ViT-B dense features from
AID, NWPU-RESISC45, and DOTA crops whose external labels match
$\mathbf{t}_{k}$. The negative set $\mathcal{N}_{k}$ is sampled
class-balanced from non-matching labels. We fit
$(\mathbf{w}_{k},b_{k})$ by 5-fold logistic regression with $L_{2}$
weight $10^{-3}$ on $\mathcal{P}_{k}\cup\mathcal{N}_{k}$, and the
probe is frozen for all downstream evaluations.

\subsection{Hyperparameters}\label{app:hyper}

Table~\ref{app:tab:hyper} summarises every hyperparameter used. None
is tuned per sensor; the same configuration is shared across the
four datasets and three seeds.

\begin{table}[h]
\centering
\small
\setlength{\tabcolsep}{4pt}
\caption{Hyperparameters shared across all datasets and seeds.}
\label{app:tab:hyper}
\begin{tabular}{ll}
\toprule
Component & Value \\
\midrule
Centroid count $K$ & $200$ \\
Router depth $R$ & $3$ \\
Routing thresholds $(\alpha_{1},\alpha_{2},\alpha_{3})$ & $(0.7,0.5,0.3)$ \\
IoU threshold $\theta_{\text{IoU}}$ & $0.5$ \\
Membership temperature $\tau$ & $0.07$ \\
Filter quantiles $(z_{1},z_{2})$ & $(0.4,0.2)$ \\
Loss weights & $(1.0,0.5,0.3,0.2)$ \\
Optimiser & AdamW \\
Learning rate & $5{\times}10^{-4}$ \\
Batch size & $64$ \\
Epochs (with halving every $200$) & $1000$ \\
Centroid $k$-means seed / restarts & $42$ / $10$ \\
Co-occurrence radius (patches) & $3$ \\
\bottomrule
\end{tabular}
\end{table}

\subsection{Per-module compute cost}\label{app:cost}

Table~\ref{app:tab:compute} reports per-module compute on a single
$64{\times}64$ PAN input on an RTX~4090. The frozen SigLIP-2 forward
pass dominates inference time but is amortised across retrieval and
tagging in any host knowledge system that already deploys a VLM.
Peak \emph{allocated} CUDA memory during training (AMP/fp16,
batch $64$, $64{\times}64$ PAN inputs, AdamW states and gradients
included) is $11.4$\,GB, decomposed disjointly into $6.8$\,GB of
Stage-I frozen-VLM activations, $2.1$\,GB of \emph{additional}
Stage-III re-encoding activations, and $2.5$\,GB for the
lightweight path, optimiser states, and buffers.
Stages~III--IV are used during training and are not required for
generation-only deployment (peak $1.9$\,GB, batch $1$); they are
invoked only when an output coverage certificate is requested,
which adds one encoder pass and $0.6$\,GB transiently.
Peak \emph{allocated} CUDA memory during training (AMP/fp16,
batch $64$, $64{\times}64$ PAN inputs, AdamW states and gradients
included) is $11.4$\,GB, decomposed disjointly into $6.8$\,GB of
Stage-I frozen-VLM activations, $2.1$\,GB of \emph{additional}
Stage-III re-encoding activations, and $2.5$\,GB for the
lightweight path, optimiser states, and buffers.
Stages~III--IV are used during training and are not required for
generation-only deployment (peak $1.9$\,GB, batch $1$); they are
invoked only when an output coverage certificate is requested,
which adds one encoder pass and $0.6$\,GB transiently.

\begin{table}[h]
\centering
\small
\caption{Per-image cost decomposition ($64{\times}64$ PAN, RTX
4090).}
\label{app:tab:compute}
\resizebox{\columnwidth}{!}{%
\begin{tabular}{lrrr}
\toprule
Module & Params & FLOPs & Time (ms) \\
\midrule
SigLIP-2 (frozen) & 400\,M & 12.1\,G & 8.2 \\
Pseudo-RGB projection & 0.01\,M & 0.003\,G & 0.1 \\
Explicit branch & 0.05\,M & 0.12\,G & 0.8 \\
Implicit router & -- & 0.09\,G & 2.4 \\
Adaptive filter & 0.03\,M & 0.04\,G & 0.5 \\
U-Net generator & 0.28\,M & 1.23\,G & 4.3 \\
Stage~III re-encode & 400\,M & 12.1\,G & 8.2 \\
Soft-Jaccard verifier & -- & 0.06\,G & 1.0 \\
\midrule
\textbf{Non-VLM lightweight path} & \textbf{0.39\,M} & \textbf{1.54\,G} & \textbf{9.1} \\
\textbf{Generation-only (Stages I--II)} & \textbf{400.4\,M} & \textbf{13.58\,G} & \textbf{16.3} \\
\textbf{Full closed loop (Stages I--IV)} & \textbf{400.4\,M} & \textbf{25.74\,G} & \textbf{25.5} \\
\bottomrule
\end{tabular}
}
\end{table}

The \emph{non-VLM lightweight path} row aggregates all non-VLM
modules; it includes the parameter-free router and verifier, and
$0.39$\,M is its trainable parameter count. Parameter sharing and
compute sharing are distinct: the frozen encoder's parameters are
counted once (Stages~I and III share weights), but its two forward
passes process different images (input versus generated output),
so both passes are counted in the FLOPs and latency of the full
closed loop.

\begin{figure}[!t]
\centering
\includegraphics[width=\columnwidth]{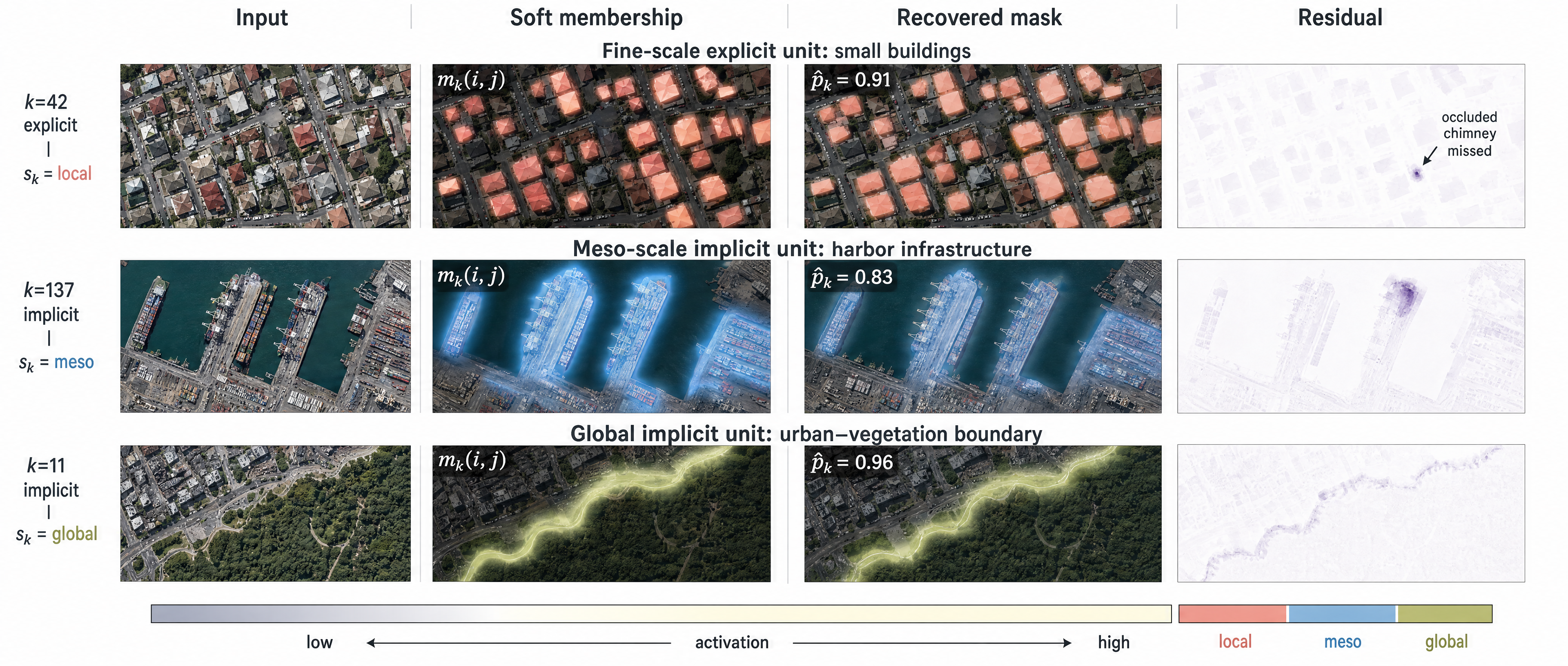}
\caption{Explicit and implicit unit activations from Stage~I on a
WV-III scene. Each row shows the membership map of a single unit
overlaid on the input tile. Implicit units (right) span relational
structures invisible to single-scale conditioning.}
\label{app:fig:activations}
\end{figure}

\subsection{Unit activations}\label{app:activations}

Figure~\ref{app:fig:activations} visualises representative explicit
and implicit unit activations on a WV-III scene. Explicit units form
compact, high-confidence supports over object-like primitives.
Implicit units span elongated or compound layouts such as
runway--apron corridors and harbour--dock--vessel relations. The
median support area is $0.6\%$ for $s_{k}{=}1$, $4.3\%$ for
$s_{k}{=}2$, and $24.7\%$ for $s_{k}{=}3$. Each row corresponds to
a single unit at one of the three scale tags ($s_{k}{=}1/2/3$); the
soft membership column shows the differentiable presence mass
$m_{k}(i,j)$ and the recovered mask column shows the verifier
response $\tilde{\mathbf{M}}_{k}$ after re-encoding the generated
image. The residual column highlights pixels where the recovered
mask disagrees with the input membership and is the signal that
$\mathcal{L}_{\text{cov}}$ minimises during training.

\section{Cross-Domain Pilot: Histopathology Super-Resolution}
\label{app:pathology}

The introduction motivates semantic collapse with whole-slide
pathology, whose concept hierarchy (cell nuclei $\to$ glandular
structures $\to$ tumour--stroma layout) mirrors the
local/meso/global pyramid. To test whether the \our recipe
transfers beyond remote sensing, we run a deliberately small pilot
on $4\times$ histopathology super-resolution.

\textbf{Setup.} We extract $8{,}000$ training and $1{,}000$ test
$256{\times}256$ tiles at $0.5\,\mu$m/px from Camelyon16 whole-slide
images, downsampling to $2\,\mu$m/px inputs. Following the recipe
in Sec.~\ref{sec:method}, only the $K{=}200$ centroids are
re-fitted (on $20$k pathology patches); \emph{every} hyperparameter
of Table~\ref{app:tab:hyper} --- $\tau$, $R$, $\alpha_{1:3}$,
$\theta_{\text{IoU}}$, loss weights, warm-up schedule --- is reused
unchanged. Prompt templates keep their schema with domain nouns
(e.g.\ ``$n_k$ instances of \{mitotic figure\}''). The independent
probe is a DINOv2 linear probe fit on PCam tumour/normal and
NCT-CRC-HE-100K tissue-type labels, never seen in training. A
$30$-query bank covers $10$ cellular, $10$ glandular, and $10$
tissue-level queries.

\begin{table}[h]
\centering
\small
\setlength{\tabcolsep}{3.5pt}
\caption{Histopathology $4\times$ SR pilot (Camelyon16 test
tiles).}
\label{app:tab:pathology}
\begin{tabular}{lcccc}
\toprule
Method & PSNR$\uparrow$ & SSIM$\uparrow$ &
$\mathcal{C}^{\text{DINO}}\!\uparrow$ & R@5$\uparrow$ \\
\midrule
Bicubic & 27.8 & 0.86 & 0.58 & 0.44 \\
SwinIR-light & 31.2 & 0.91 & 0.66 & 0.55 \\
Open-loop SigLIP U-Net & 31.4 & 0.91 & 0.70 & 0.58 \\
\rowcolor{blue!4}\our (closed-loop) & \textbf{31.6} & \textbf{0.92}
& \textbf{0.83} & \textbf{0.71} \\
\bottomrule
\end{tabular}
\end{table}

\textbf{Results.} Table~\ref{app:tab:pathology} reproduces the
remote-sensing pattern: pixel metrics are nearly saturated across
learning-based methods ($31.4$ vs.\ $31.6$ PSNR), yet closed-loop
coverage verification lifts label-grounded coverage from $0.70$ to
$0.83$ and mean Recall@5 from $0.58$ to $0.71$ ($+13$\,pp). The
gain concentrates on cellular-scale queries ($0.40\!\to\!0.58$
R@5), the pathology analogue of the small-object result of
Table~\ref{app:tab:dota_perclass}. The pilot is small and uses a
single organ site, so we do not claim clinical validity; its role
is to show that the closed-loop recipe, with unchanged
hyperparameters and only re-fitted centroids, is not specific to
satellite imagery.

\end{document}